\documentclass[aps,pra,superscriptaddress,twocolumn]{revtex4-1}

\usepackage{gensymb}
\usepackage{amsmath}
\usepackage{amssymb}
\usepackage{bm} 
\usepackage{color}
\usepackage[linkcolor=blue, colorlinks=true, urlcolor=blue,
citecolor=red]{hyperref}
\usepackage{graphicx}
\usepackage[caption=false]{subfig}
\usepackage{bbold}
\usepackage{verbatim}
\usepackage{appendix}
\usepackage{amsthm}
\usepackage[normalem]{ulem}

\theoremstyle{plain} 
\newtheorem{thm}{Theorem} 
\newtheorem{prop}{Proposition} 
\newtheorem{approp}{Proposition}[section]
\newtheorem{aplem}{Lemma}[section]

\newcommand{\ket}[1]{\left| #1 \right>} 
\newcommand{\bra}[1]{\left< #1 \right|} 

\newcommand{\ketbras}[3]{\ket{#1}_{#3}\hspace*{-0.mm}\bra{#2}}

\newcommand{\Id}{\mathbb{1}}
\newcommand{\bId}{{\boldsymbol 1}}
\newcommand{\Tr}{\mathrm{\text{Tr}}}
\newcommand{\aop}{\hat{a}}
\newcommand{\bop}{\hat{b}}
\newcommand{\rop}{{\bf \hat{r}}}
\newcommand{\xop}{\hat{x}}
\newcommand{\pop}{\hat{p}}
\newcommand{\bOmega}{{\bf \Omega}}
\newcommand{\bve}{{\pmb v}}
\newcommand{\bxi}{{\pmb \xi}}
\newcommand{\bsigma}{{\pmb \sigma}}
\newcommand{\beeta}{{\pmb \eta}}
\newcommand{\bGamma}{{\bf \Gamma}}
\newcommand{\bz}{{\pmb 0}}

 \newcommand{\bbeta}{{\pmb \beta}}
 \newcommand{\bu}{{\pmb u}}
 
\begin{document}

\title{Gaussian Discriminating Strength}

\author{L. Rigovacca}
\affiliation{QOLS Group, Imperial College London, Blackett Laboratory, SW7 2AZ London, UK}   
 
\author{A. Farace}
\affiliation{NEST, Scuola Normale Superiore and Istituto Nanoscienze-CNR, I-56126 Pisa, Italy} 
\affiliation{Max-Planck-Institut f\"ur Quantenoptik, Hans-Kopfermann-Stra{\ss}e 1, 85748 Garching, Germany} 

\author{A. De Pasquale}
\affiliation{NEST, Scuola Normale Superiore and Istituto Nanoscienze-CNR, I-56126 Pisa, Italy} 

\author{V. Giovannetti}
\affiliation{NEST, Scuola Normale Superiore and Istituto Nanoscienze-CNR, I-56126 Pisa, Italy}

\begin{abstract}
We present a quantifier of non-classical correlations for bipartite, multi-mode Gaussian states. It is derived from the Discriminating Strength measure, introduced for finite dimensional systems  in A. Farace et al., New. J. Phys. {\bf 16}, 073010 (2014).  As the latter the new measure exploits the Quantum Chernoff Bound to gauge the  susceptibility of the composite system with respect to local perturbations induced by unitary gates extracted from a suitable set of allowed transformations (the latter being identified by posing some general requirements). Closed expressions are provided for the case of two-mode Gaussian states obtained by squeezing or by linearly mixing via a beam-splitter a factorized two-mode thermal state. For these density matrices, we study how non-classical correlations are related with the entanglement present in the system and with its  total photon number.  
\end{abstract}
\maketitle

\section{Introduction}\label{sec: Intro}
The presence of correlations in bipartite quantum systems, that cannot be explained by means of classical probability distributions, is one of the main features of quantum mechanics. The most  important among such correlations is surely entanglement \cite{Horodecki_EntReview}. However, recently much attention has been devoted to the study and the characterization of correlations that go beyond the paradigm of entanglement, being necessary but not sufficient for its presence. They are generically known under the name of discord-like correlations \cite{Modi_Review}, from the name of the first quantifier introduced to measure them, i.e. Quantum Discord \cite{Zurek_QDiscord,Vedral_QDiscord},
and appear when the state of a (say) bipartite system $AB$ cannot be written as a mixture of separable terms that could be locally distinguished in one (or both) subsystems. 
Depending on this last choice the set of states with zero Discord is composed by Classical-Quantum (CQ) states, when local distinguishability is required for example only on subsystem $A$, or Classical-Classical (CC) states, when it is required on both subsystems: 
\begin{align}
\rho^{(CQ)}_{AB}&=\sum_{i} p_{i} \ketbras{i}{i}{A}\otimes\rho_B^{(i)}, \label{def: CQ state}\\
\rho^{(CC)}_{AB}&=\sum_{i,j} p_{ij} \ketbras{i}{i}{A}\otimes\ketbras{j}{j}{B}\label{def: CC state},
\end{align}
where $\{\ket{i}_A\}_i$ and $\{\ket{j}_B\}_j$ are orthonormal local bases (note that differently from the totality of separable states, these sets are not convex).

Since the introduction of the original Quantum Discord, many other measures $\mathcal M$ have been introduced to quantify the amount of non-classical correlations present in a state, approaching the problem from different perspectives. All of them, however, must satisfy a set of conditions on which the community agrees:
\begin{enumerate}
	\item $\mathcal{M}\geq 0$ with equality on CQ (or CC) states;
	\item $\mathcal{M}$ is invariant under local unitary operations;
	\item $\mathcal{M}$ reduces to an entanglement monotone on pure states;
\end{enumerate}
with an additional requirement for an asymmetric measure testing the classicality of subsystem $A$ (i.e. defined with respect to CQ states):
\begin{enumerate}
	\setcounter{enumi}{3}
	\item $\mathcal{M}$ is monotonically decreasing under CPTP maps on the untested subsystem $B$. 
\end{enumerate}

A first class of measures uses a geometric description, based on the minimization of the distance between the considered state and the set of CQ (or CC) states \cite{Horodecki_QDeficit,Modi_RelEntropy,Orszag_BuresGeometric,Orszag_BuresGQubit,Sarandy_Geom1Norm,Giovannetti_TDDQubit,Adesso_NoQ,Piani_GDiscordProblem}.
 Another possible approach is instead rooted on the key observation that any state of the form \eqref{def: CQ state} is left invariant by a (non-trivial) opportunely chosen local measurement or local unitary operation acting on the first subsystem. On the other side, if there exists such local action that leaves the state unchanged, then the state must be in the CQ set. This suggests a way to quantify the asymmetric non-classicality of a quantum state based on the minimum change induced on the whole state by a local action, in the form of either a measurement \cite{Luo_Measurement,Luo_MID,Adesso_GaussMID} or a unitary operation. Focusing on the latter, this minimum change can be quantified by
\begin{equation}\label{def: DoR}
\mathcal{M}_A^{(\mathcal S)}(\rho_{AB}) = \min_{U_A\in \mathcal{S}} D\left(\rho_{AB}, U_A\right),
\end{equation}
where $D$ is some well-behaved non-negative functional depending on the state and on a local unitary $U_A$, and $\mathcal S$ is a suitably chosen set of unitary operations. 
Let us point out that the idea of characterising quantum features of a bipartite system through the detection of global alterations induced by local unitary operations is not new. For example, in \cite{Fu_NonLocality} it is shown that global alterations due to local cyclic operations are a signature of a particular	form of non-locality which is beneficial in	superdense coding schemes. However, as already observed in \cite{Datta_Inequivalence}, the latter investigation is only marginally related to the study of discord-like quantum correlations, which instead is the goal of the present work. In this respect, our assumptions on the set of allowed unitary operations and the subsequent minimization	process will greatly differ from what is done in \cite{Fu_NonLocality}. 

Expression \eqref{def: DoR} lies at the heart of several measures of correlations known in the literature: the Interferometric Power (IP) \cite{Girolami_IP} if $D$ is proportional to the Quantum Fisher Information \cite{Paris_QFI} of the state when the phase to be estimated is encoded by means of $U_A$, the Discriminating Strength (DS) \cite{DS} if $D$ is related to the Quantum Chernoff Bound \cite{Audenaert_QCB} involved in the discrimination between the state and its evolution under $U_A$, and the Discord of Response (DoR) \cite{Illuminati_DoR} if $D$ represents a contractive distance between such states.
The main idea underlying this approach is to exploit the effect of the global alteration, that in presence of correlations is necessarily induced by any local unitary map, in order to have an advantage in some information processing task, where not all the details characterizing the local evolution are known in advance. From this perspective all of the aforementioned measures have a clear operational meaning, being strictly related to the performance of some tasks: phase estimation for the IP, state discrimination for the DS and quantum reading for the DoR. It is also worth stressing that, while being all proper measures of discord-like correlations, the quantities introduced in~\cite{Girolami_IP,Paris_QFI,DS,Audenaert_QCB,Illuminati_DoR} are inequivalent as no proportionality factor can in general be found among them (the different choice of the functional $D$ being associated with different potential applications of the state $\rho_{AB}$).

Let us emphasize that the minimization set $\mathcal{S}$ in Eq.~\eqref{def: DoR} plays an important role in guaranteeing that $\mathcal{M}_A^{(\mathcal S)}(\rho_{AB})$ is a proper measure of discord-like correlations. It is trivial to see that $\mathcal S$ cannot be taken to be the whole unitary group, otherwise choosing $U_A=\Id_A$ the corresponding quantifier would be zero for all states. It has been shown that condition $1.$ (i.e. $\mathcal{M}=0$ exactly on the CQ set) is satisfied for all these measures if the set $\mathcal S$ is taken to be the subset of the unitary operations with a fixed non-degenerate spectrum \cite{Girolami_IP,DS}. Therefore one can introduce a different well-defined quantifier for each spectrum choice. Selecting the optimal one among all possible choices is far from a trivial task. 
The symmetric spectrum composed by the $d_A=\dim \mathcal{H}_A$ roots of unity is conjectured to offer the best resolution, but there is no formal proof yet. Eventually, let us point out that although it is possible that other choices of $\mathcal S$ (apart from fixing a non-degenerate spectrum and spanning all possible basis) might lead to a proper quantifier of discord-like correlations, none is known so far.

Although the previous discussion holds for arbitrary systems, most of the quantities involved were first explicitly computed only for finite-dimensional systems, often qubits. However, there are several cases of interest, for instance in the context of optical interferometry, where one might be interested in evaluating the correlations present in states of continuous-variable systems, especially Gaussian states~\cite{Ferraro_Review,Weedbrook_GaussianReview,Wang_Review}.
In order to answer this question, some measures of non-classical correlations have already been reformulated by restricting the corresponding operations (measurements or quantum maps) so as to preserve the Gaussian character of the considered  states \cite{Paris_GaussianDiscord,Adesso_GaussianDiscord,Adesso_GaussianDiscordOld}. 
However, since this involves a minimization over a restricted set of operations, such quantities will in general be upper bounds for their original counterparts. Only recently the resolution of the Gaussian minimum-entropy conjecture \cite{Giovannetti_GaussConjecture1, Giovannetti_GaussConjecture2} allowed to conclude that the Gaussian Quantum Discord \cite{Paris_GaussianDiscord,Adesso_GaussianDiscord,Adesso_GaussianDiscordOld} coincides with the original Discord on a vast class of states \cite{Pirandola_OptimalityGaussDiscord}.

This manuscript is developed within the mindset of Gaussian quantum states.
In Sec.~\ref{sec: Unitary choice} we first review some basic notions about Gaussian states and Gaussian unitary transformations, then we address the problem of characterizing the set $\mathcal{S}$ within the Gaussian scenario previously detailed. We start our discussion with the standard choice of imposing a fixed spectrum for the unitary operations in $\mathcal S$ with no restrictions on the basis. Next we show how the non-degeneracy condition is changed when we move from the most general case of arbitrary states and operations to the Gaussian scenario.
In particular, if we require some basic properties for the functional $D\left(\rho_{AB}, U_A\right)$, not all spectra among those associated with Gaussian unitary operations can be chosen in order for $\mathcal{M}_D^{(\mathcal S)}(\rho_{AB})$ to be a proper measure of discord-like correlations. Interestingly, the class of allowed sets $\mathcal{S}$ for two-mode systems reduces to a form that has been assumed a priori in the specification  of some previous measures of non-classical correlations to the Gaussian setting \cite{Adesso_GaussianIP,Illuminati_GaussianDoR}, and that now finds a formal justification.
We will then focus on the Gaussian version of the  DS \cite{DS}, based on the Quantum Chernoff bound, for which an explicit analysis in the context of continuous variable systems is still missing. In Sec.~\ref{sec: Gaussian DS} we will obtain a formal expression for the Gaussian DS and we will explicitly evaluate it for some classes of two-mode states, obtained through linear mixing or two-mode squeezing applied on thermal states. We will also discuss the relations between Gaussian DS and entanglement, squeezing and total number of photons for these states, drawing our conclusions in Sec.~\ref{sec: Conclusions}.

\section{Choice of local Gaussian unitary}\label{sec: Unitary choice}

Let us consider $n$ bosonic modes of a continuous-variable system, described by the annihilation operators $\{\aop_i\}_{i=1}^n$. It will be useful to reorganize them within the quadrature vector, satisfying the canonical commutation relation:
\begin{equation}\label{def: r operator}
\rop=(\xop_1,\pop_1,\dots,\xop_n,\pop_n)^\intercal, \qquad \quad  [\rop,\rop^\intercal]=i \bOmega_n,
\end{equation}
where $\xop_i=(\aop_i+\aop_i^\dagger)/\sqrt{2}$, $\pop_i=-i(\aop_i-\aop_i^\dagger)/\sqrt{2}$. The matrix $\bOmega_n$ in Eq.~(\ref{def: r operator})  is the $n$-mode symplectic form
\begin{equation}\label{def: symplectic form}
\bOmega_n =\bigoplus_{i=1}^n\left(\begin{array}{cc}
0 & 1\\
-1 & 0
\end{array}\right),
\end{equation}
which enters in the Robertson-Schr\"{o}dinger uncertainty relation that  any density matrix $\rho$ of the system has to satisfy, i.e. 
\begin{equation}\label{eq: uncertainty on Gamma}
	\bGamma + i\,\bOmega_n \geq 0,
\end{equation}
with  $\bGamma$ being the covariance matrix of $\rho$, 
\begin{equation}\label{def: Gauss moments2}
\bGamma=\Tr\left[\rho\, \{\rop - \bxi,\rop^\intercal-\bxi^\intercal\}_+\right].
\end{equation}
Here $\{\cdot,\cdot\}_+$ describes the anti-commutator and  
\begin{equation}\label{def: Gauss moments1}
\bxi = \Tr[\rho \, \rop], 
\end{equation}
is  the associated  displacement vector~\cite{Weedbrook_GaussianReview,Ferraro_Review,Wang_Review}.

In what follows we will focus on the set $\mathfrak{G}$ formed by the Gaussian states, i.e. by the density matrices $\rho$  which can be expressed as Gibbs states  of Hamiltonians at most quadratic in the quadrature operator $\rop$. They have a Gaussian characteristic function and are hence completely determined by the first and second moments of the operator $\rop$, i.e. 
by the vector $\bxi$ of Eq.~(\ref{def: Gauss moments1}) and by the matrix $\bGamma$ of Eq.~(\ref{def: Gauss moments2}). Directly associated with the notion of Gaussian states, is the set 
 $\mathcal{G}_n$ formed by  $n$-mode Gaussian unitary transformations, i.e. by those unitary operations which maps $\mathfrak{G}$ into itself. 
 By construction, the element of  $\mathcal{G}_n$ must 
 induce a Heisenberg evolution of $\rop$ which is linear, while  ensuring that the commutation relation \eqref{def: r operator} holds   for the evolved quadrature vector operator as well. These two conditions imply that any $U \in \mathcal{G}_n$  can be unambiguously represented   as 
\begin{equation}\label{eq: Gauss unitary action}
U^\dagger\,\rop\, U = {\bf U}\, \rop + \beeta_U, 
\end{equation}  
where $\beeta_U$ is a vector of $\mathbb{R}^{2n}$ and  ${\bf U}$ is a matrix belonging to the symplectic group $Sp(2n,\mathbb{R})$, formed by the $2n\times 2n$ real matrices ${\bf S}$ that preserves the symplectic form $\bOmega_{n}$ \eqref{def: symplectic form}, i.e. 
\begin{equation}
{\bf S}\in Sp(2n,\mathbb{R}) \ \Leftrightarrow \ {\bf S}\,\bOmega_{n}\, {\bf S}^\intercal = \bOmega_n.
\end{equation} 
From Eqs.~(\ref{def: Gauss moments2}) and (\ref{def: Gauss moments1}) one can easily verify that in terms of the displacement $\bxi$ and the covariance matrix $\bGamma$, Eq.~(\ref{eq: Gauss unitary action}) induces the following 
transformations
\begin{equation}\label{eq: Gaussian action on moments}
\bxi \rightarrow {\bf U}\,\bxi + \beeta_U, \qquad \qquad \bGamma \rightarrow {\bf U}\, \bGamma \, {\bf U}^\intercal,
\end{equation}
which fully characterize the input-output mapping at the level of Gaussian states $\mathfrak{G}$. 
A  special class of Gaussian unitaries which will play an important role in the remaining of the paper is provided by the phase-transformations   
\begin{equation}\label{eq:allowed spectra}
U= e^{i\sum_{j=1}^{n}\lambda_j \aop_j^\dagger\aop_j},
\end{equation}
 with $\lambda_j$  real parameters. In the representation~(\ref{eq: Gauss unitary action}) they
are characterized by a null vector $\beeta_U=\bz$ and 
by a block-form  symplectic matrix  $\oplus_{j=1}^{n} {\bf R}(\lambda_j)$, with  ${\bf R}(\lambda)\in SO(2)$ being the rotation
\begin{equation}\label{def: R}
{\bf R}(\lambda)=\left(\begin{array}{cc}
\cos\lambda & -\sin\lambda \\
\sin\lambda & \cos\lambda
\end{array}\right).
\end{equation}
Let us finally recall  that according to the Williamson decomposition~\cite{Williamson,Weedbrook_GaussianReview,Ferraro_Review,Wang_Review} given the covariance matrix $\bGamma$ of an arbitrary state $\rho$ there exists a symplectic matrix
${\bf S}\in Sp(2n,\mathbb{R})$ and a set of coefficients $\{\nu_j\}_{j=1}^n$ such that
\begin{equation}\label{def: Williamson decomposition}
\bGamma = {\bf S} \left(\bigoplus_{j=1}^n \nu_j \bId_2\right){\bf S}^{\intercal}.
\end{equation}
The coefficients $\{\nu_j\}_{j=1}^n$, which are called \textit{symplectic eigenvalues} of $\bGamma$, can be computed as the regular eigenvalues of the matrix  $|i \bOmega_n\bGamma|$,  and fulfill the inequality $\nu_j \geq 1$ due to constraint \eqref{eq: uncertainty on Gamma}. In particular this bound is saturated on pure states, when all $\nu_{j}$ are equal to $1$. The class of Gaussian states admitting a Williamson decomposition with ${\bf S}=\mathbb{1}$  are usually denominated \textit{thermal states}.
\\

We have now all the elements to present our first result. For this purpose let us assume that the $n$ modes of the system are split in two sets: the set $A$ containing $n_A$ modes controlled by Alice and the set $B$ with the remaining $n_B = n- n_A$ modes, controlled by Bob. We are interested in constructing 
a  quantifier $\mathcal M_A^{(\mathcal S)}$ of the form  \eqref{def: DoR}, which can be used to characterize  the correlations between $A$ and $B$, by restricting the analysis to the case where the allowed initial states $\rho_{AB}$ are Gaussian density matrices of the joint system $AB$, and under the additional constraint of reducing the set $\mathcal S$ to a proper subset of the Gaussian unitary transformations $\mathcal{G}_{n_A}$, which operates locally on $A$. As briefly mentioned in the introduction the last assumption is in general not justified a priori: analogously to what done in Refs.~\cite{Paris_GaussianDiscord,Adesso_GaussianDiscord,Adesso_GaussianDiscordOld}
 it is only motivated  by the need of simplifying the analysis by forcing also the transformed counterpart of $\rho_{AB}$ under $U_A$ to be Gaussian.
We will also assume to deal only with sets $\mathcal S$ composed by all Gaussian unitary operations with a certain fixed spectrum, since this is the only known choice that in a generic (non-Gaussian) framework leads to a measure of discord-like correlations. The following Theorem states that under a few basic and reasonable assumptions the spectrum characterizing the set $\mathcal S$ must be chosen among those associated with non-trivial phase-transformations that act locally on every mode, i.e. with Gaussian unitary operations as in Eq. \eqref{eq:allowed spectra} with $\lambda_j$	not being integer multiplies of $2\pi$. This can be interpreted as the Gaussian counterpart of the non-degeneracy condition that applies in a generic framework where no Gaussian restrictions are imposed.
 
\begin{thm} \label{thm: local rotation}
	Let us consider a non-negative functional $D(\rho_{AB},U_A)$, depending upon a Gaussian state $\rho_{AB}$ and a local Gaussian unitary $U_{A}$, such that: 
	\begin{enumerate}
		\item[D1.] $D(\rho_{AB},U_A)=0 \;\Longleftrightarrow \;\rho_{AB} = U_A\,\rho_{AB}\,U_A^\dagger$,
		\item[D2.] $D(V_B\,\rho_{AB}\,V_B^\dagger,U_A)=D(\rho_{AB},U_A), $ for all local Gaussian unitary operations $V_B \in \mathcal{G}_{n_B}$, 
		\item[D3.] $D(V_A\,\rho_{AB}\,V_A^\dagger,U_A)=D(\rho_{AB},V_A^\dagger\,U_A\,V_A)$, for all local Gaussian unitary operations $V_A \in \mathcal{G}_{n_A}$.
	\end{enumerate}
	 Let us also suppose that given $W_A \in \mathcal{S}$, the whole set can be written as
	 \begin{equation}\label{eq: U_A - W_A decomposition}
 \mathcal{S} = \left\{U_A: U_A=V_A W_AV_A^\dagger, V_A \in \mathcal{G}_{n_A}\right\}.
\end{equation} 
	
Then the associated quantity $\mathcal{M}_A^{(\mathcal S)}(\rho_{AB})$, defined as in Eq.~(\ref{def: DoR}), 
satisfies  the property
	\begin{itemize}
		\item[M1.] ${\cal M}_A^{(\mathcal S)}(U_{\text{loc}} \,\rho_{AB}\,U_{\text{loc}}^\dagger)={\cal M}_A^{(\mathcal S)}(\rho_{AB})$, for all local unitary Gaussian operations $U_{\text{loc}} = V_A\otimes V_B$, with $V_A\in \mathcal{G}_{n_A}$ and $V_B\in \mathcal{G}_{n_B}.$
		\end{itemize} 
		Furthermore, it will also fulfill the  property
		\begin{itemize} 
		\item[M2.] ${\cal M}_A^{(\mathcal S)}(\rho_{AB})=0 \;\Longleftrightarrow \; \rho_{AB}=\rho_A\otimes \rho_B$, 
	\end{itemize}
	if and only if the operator $W_A$ of Eq.~(\ref{eq: U_A - W_A decomposition}) can be identified with a non trivial phase-transformation $e^{i\sum_{j=1}^{n_A}\lambda_j \aop_j^\dagger\aop_j}$,  i.e. if and only if the set $\mathcal S$ writes as
\begin{equation}\label{def: form S lambda}
	\mathcal{S}_{\{\lambda_{j}\}} = \left\{U_A: U_A=V_A e^{i\sum_{j=1}^{n_A}\lambda_j \aop_j^\dagger\aop_j}V_A^\dagger, V_A \in \mathcal{G}_{n_A}\right\},
\end{equation} 
	where all the real parameters $\lambda_{j}$ are required not to be integer multiples of $2\pi$.	
\end{thm}

Before going into the details of the proof, let us comment why the assumptions $D1$--$D3$, the structural assumption~(\ref{eq: U_A - W_A decomposition}), 
 and the properties $M1$, $M2$ appear quite natural. 
 \begin{itemize} 
\item The set $D1$--$D3$  refers to the ``distance-like'' quantifier $D$ and are the Gaussian counterparts of analogous properties which are satisfied 
by all the measures of the form~(\ref{def: DoR}) introduced so far, i.e. 
 the Interferometric Power~\cite{Girolami_IP}, the Discriminating Strength~\cite{DS}, and the Discord of Response~\cite{Illuminati_DoR}. In particular  $D1$ requires that $D$ should be sensitive to all those unitary operations that alter the considered state. If this were not the case, $\mathcal{M}^{(\mathcal S)}_D$ would not be able to quantify properly the susceptibility of $\rho_{AB}$ under local unitary maps, on which these quantifiers of non-classical correlations rely upon. $D2$ and $D3$ derive from the reasonable condition
$D(\rho_{AB},U_A) = D(V\rho_{AB}V^\dagger , VU_A V^\dagger)$,   
that describes the independence of the functional $D$ upon a change of basis.
\item The hypothesis~\eqref{eq: U_A - W_A decomposition} on the structure of the elements of $\mathcal S$ can be interpreted similarly to the requirement of having a fixed spectrum in the finite-dimensional case \cite{Girolami_IP,Illuminati_DoR,DS}, without restrictions on the basis. 
 \item $M1$, $M2$ refer instead to the properties that a good measure of non-classical correlations $\mathcal{M}_D^{(\mathcal S)}$ is expected to satisfy. In particular, $M1$ is just the invariance under local unitary operations (see $2.$ in Sec. \ref{sec: Intro}) applied to the Gaussian setting, while $M2$ is nothing but condition $1.$ of Sec.~\ref{sec: Intro}. 
Indeed, in \cite{Rahimi-Keshari_DiscordVerification} and \cite{Adesso_CQGaussian} is shown that the only CQ [see Eq. \eqref{def: CQ state}] Gaussian states are those that are completely uncorrelated, respectively in the two-mode and multi-mode case. A rederivation of this same fact can be found for completeness in Appendix \ref{app: CQ gaussian states}. The measures $\mathcal{M}_D^{(\mathcal S)}$ can therefore be rigorously interpreted as Gaussian multi-mode quantifiers of discord-like correlations, since they nullify exactly on the same set where Quantum Discord does \cite{Zurek_QDiscord,Vedral_QDiscord}.
\end{itemize}

\begin{proof}[Proof of Theorem 1.]
The property $M1$ follows trivially from $D2$ and $D3$ and from Eq.~(\ref{eq: U_A - W_A decomposition}).
Let us next prove that the form of $\mathcal{S}_{\{\lambda_j \}}$ is necessary to get $M2$.
To do so, let us consider an uncorrelated initial Gaussian state ${ \rho_A\otimes\rho_B }$, parametrized by a covariance matrix $\bGamma_{AB}=\bGamma_A\oplus\bGamma_B$ and a displacement vector $\bxi$.
Williamson decomposition \eqref{def: Williamson decomposition}, together with $M1$, allows us to assume without loss of generality 
\begin{equation}\label{eq: product state easy form}
\bGamma_A =\bigoplus_{j=1}^{n_A}\nu_j \bId_2, \qquad \bxi=\bz,
\end{equation}
with $\nu_j\geq 1$ $\forall j=1,\ldots,n_A$.
For every product state, requirements $M2$ and $D1$ impose the existence of a particular $W_A \in \mathcal {S}$ that preserves it. 
From the relation
\begin{align}
	W_A^\dagger\, \rop \, W_A &= \left({\bf W}_A\oplus\bId_B\right) \,\rop + \left(\beeta_W^{(A)},\bz^{(B)}\right)^\intercal,\label{eq: U_A action}
\end{align}
such invariance condition imposes
that $ \beeta_W^{(A)}=\bz^{(A)}$ and 
\begin{equation}
{\bf W}_A\left(\bigoplus_{j=1}^{n_A}\nu_j\bId_2\right){\bf W}_A^\intercal = \left(\bigoplus_{j=1}^{n_A}\nu_j\bId_2\right),
\end{equation}
which, by setting  $\nu_j=1$ $\forall j$, in particular  implies   
\begin{equation}
{\bf W}_A\in Sp(2n_A)\cap O(2 n_A).
\end{equation}
In Appendix \ref{app: symp decomposition} it is shown  that any such ${\bf W_A}$ can be transformed into a direct sum of single-mode rotations by means of orthogonal and symplectic matrices, which can be adsorbed into the action of $V_A$ (appearing in the assumed structure of $\mathcal S$ \eqref{eq: U_A - W_A decomposition}). In other words we can take 
\begin{eqnarray} {\bf W}_A=\bigoplus_{j=1}^{n_A} {\bf R}_A(\lambda_j)\;,  \label{FORMAWA} 
\end{eqnarray}  with ${\bf R}_A(\lambda_j)\in SO(2)$ 
identifying hence $W_A$ with the Gaussian unitary phase-transformation 
	$e^{i\sum_{j=1}^{n_A}\lambda_j \aop_j^\dagger\aop_j}$ -- see discussion above Eq.~(\ref{def: R}).
Finally, in order to have a non-trivial measure (condition $M2$) ${\bf W}_A$ must be different from the identity 
in each of the blocks, imposing hence all the $\lambda_j$'s not to be integer multiples of $2\pi$.

The converse,  i.e. proving that  property $M2$ holds when assuming $\mathcal{S}_{\{\lambda_j \}}$ as in \eqref{def: form S lambda}, follows from the fact that  the non trivial phase-transformation
$e^{i\sum_{j=1}^{n_A}\lambda_j \aop_j^\dagger\aop_j}$  admits the symplectic form~(\ref{FORMAWA}) and a null displacement vector [see Eq.~(\ref{def: R})]. Accordingly,  
via the correspondence~(\ref{eq: Gauss unitary action})  the elements $U_A$ of the set~(\ref{def: form S lambda}) are characterized by
symplectic matrices ${\bf U}_A$ and  vectors $\beeta_U^{(A)}$  of the form 
\begin{eqnarray}\label{def: UA rotation1}
{\bf U}_A&=&{\bf V}_A\,\Big( \oplus_{j=1}^{n_A} {\bf R}_A(\lambda_j)\Big) \,{\bf V}_A^{-1}, \\
\beeta_U^{(A)} &=& (\bId-{\bf U}_A)\,\beeta_V^{(A)}, \label{def: vec1} 
\end{eqnarray} 
where 
 ${\bf V}_A$ and  $\beeta_V^{(A)}$ are, respectively, the symplectic matrix and the vector of the Gaussian unitary $V_A$. 
Now let $\rho_{AB}$ be a Gaussian state which nullifies the quantity ${\cal M}_A^{(\mathcal S)}(\rho_{AB})$. From $D1$ this is only possible if there exists $U_A$ in $\mathcal{S}_{\{\lambda_j \}}$ which
leaves such state invariant. At the level of covariance matrices this formally implies the condition 
\begin{equation}\label{eq: covariance matrix blocks}
\bGamma_{AB}= 
\left(\begin{array}{c|c}
\bGamma_A & \bGamma_{OFF} \\\hline
\bGamma_{OFF}^{\intercal} & \bGamma_B
\end{array}\right)=\left(\begin{array}{c|c}
{\bf U}_A\bGamma_A{\bf U}_A^\intercal & {\bf U}_A\bGamma_{OFF} \\\hline
\bGamma_{OFF}^{\intercal}{\bf U}_A^\intercal & \bGamma_B
\end{array}\right),
\end{equation}
where $\bGamma_{AB}$ is the covariance matrix of $\rho_{AB}$ expressed in  blocks form ($\bGamma_A$ and $\bGamma_B$ being the covariance matrices of the reduced density operators associated with the subsystems $A$ and $B$ respectively),  and where
the last term refers to the joint covariance matrix of $\rho_{AB}$ after the action of $U_A$. 
By focusing on the off diagonal-blocks we observe that Eq.~(\ref{eq: covariance matrix blocks}) requires 
\begin{eqnarray}\label{SOLC}
{\bf V}_A^{-1}  \bGamma_{OFF} =\oplus_{j=1}^{n_A} {\bf R}_A(\lambda_j)  \; {\bf V}_A^{-1} \bGamma_{OFF},
\end{eqnarray}  
for some symplectic matrix ${\bf V}_A$. Notice that satisfying the above expression is equivalent to finding, for each block, a collection of  vectors $\bve$ solving the 
eigenvalue equation ${\bf R}_A(\lambda_j) \bve = \bve$.  
However, by construction ${\bf R}_A(\lambda_j)$ admits as eigenvalues the phases $e^{\pm i \lambda_j} \neq 1$. Accordingly, the only possible solution is  $\bve=0$, i.e. to have ${\bf V}_A^{-1} \bGamma_{OFF} =\bz$ in Eq.~(\ref{SOLC}), i.e. 
 $\bGamma_{OFF} =\bz$, i.e. to have that the Gaussian density matrix $\rho_{AB}$ is a product state.
\end{proof}

\section{Gaussian Discriminating Strength}\label{sec: Gaussian DS}
In the following we will focus on a particular quantifier of non-classical correlations, the Discriminating Strength (DS) introduced in~\cite{DS}, that can be defined using in Eq. \eqref{def: DoR} the following functional:
	\begin{equation}\label{def: D_DS}
	D_{DS}(\rho_{AB},U_A) = 1- Q(\rho_{AB},U_A \,\rho_{AB}\, U_A^\dagger),
	\end{equation}
where $Q$ acts on a pair of states as 	
\begin{equation}\label{def: QCB}
Q(\rho_0,\rho_1) = \min_{s\in[0,1]}\Tr\left[\rho_0^s\,\rho_1^{1-s}\right].
\end{equation}
This yields the following expression for the DS of a bipartite state $\rho_{AB}$:
\begin{eqnarray}
{\cal DS}_A^{(\mathcal S)}(\rho) &=& 1-\max_{U_A\in \mathcal{S}} Q\left(\rho, U_A\,\rho\,U_A^\dagger \right)\label{def: DS}\\
&=& 1-\max_{U_A\in \mathcal{S}} \min_{s\in[0,1]}\Tr\left[\rho^s\,U_A\,\rho^{1-s}\,U_A^\dagger\right].
\end{eqnarray}
The quantity $Q$ that appears in the above definitions is the Quantum Chernoff Bound \cite{Audenaert_QCB} that intervenes in a state-discrimination scenario. We refer to Appendix \ref{app: DS protocl} for a physical interpretation of such quantity, that provides the operational meaning of the Discriminating Strength measure.

We will now discuss the problem of obtaining an expression for a quantifier analogous to the DS for the class of Gaussian states. 
For this purpose we will rely on the results of Theorem~\ref{thm: local rotation} by observing that the functional in Eq. \eqref{def: D_DS}
fulfills the assumptions $D1$--$D3$ when evaluated over Gaussian states and Gaussian unitaries -- this can be trivially verified directly from Eq.~(\ref{def: QCB}). 
We will hence restrict the optimization set $\mathcal S$ 
in \eqref{def: DS}  as described in Sec.~\ref{sec: Unitary choice} . In particular we face the task of evaluating $Q(\rho_{AB},U_A\,\rho_{AB}\, U_A^\dagger)$ for a Gaussian state $\rho_{AB}$ and a Gaussian operation $U_A$. To do so we can use the Williamson decomposition \eqref{def: Williamson decomposition}, that allows to obtain simple expressions for the exponentiated states appearing in the Quantum Chernoff Bound \eqref{def: QCB}. Indeed, it can be shown that, up to a normalization factor \begin{equation}
\Tr[\rho_{AB}^s] =\Pi_{j=1}^{n} G_s(\nu_j),
\end{equation}
$\rho_{AB}^s$ is still a Gaussian state whose covariance matrix $\bGamma_{AB}^{(s)}$ can be obtained from \eqref{def: Williamson decomposition} just changing the original symplectic eigenvalues
$\{\nu_j\}_j$, pertaining to $\rho_{AB}$, to the set of functions $\{\Lambda_s(\nu_j)\}_j$, 
$\Lambda_s$ and $G_s$ being explicitly evaluated as \cite{Pirandola_LambdaGFunctions}:
\begin{align}\label{def: lambda function}
\Lambda_s(x)&=\frac{(x+1)^s + (x-1)^s}{(x+1)^s - (x-1)^s},\\
G_s(x)&=\frac{2^s}{(x+1)^s - (x-1)^s}.
\end{align}
Noticing also that $\rho_{AB}$ and $U_A\,\rho_{AB} \,U_A^\dagger$  show
the same symplectic eigenvalues, and exploiting the characteristic function formalism \cite{Marian_Fidelity},
the Quantum Chernoff Bound \eqref{def: QCB} can then be evaluated as \cite{Illuminati_GaussianDoR}:
\begin{equation}\label{eq: Formal QCB}
Q(\rho_{AB},U_A\rho_{AB} U_A^\dagger)=\min_{s\in[0,1]}Q_s e^{-\Delta_s},
\end{equation}
with
\begin{align}
Q_s &= \frac{\Pi_{j=1}^{n} \left[\Lambda_s(\nu_j) + \Lambda_{1-s}(\nu_j) \right]}{\sqrt{\det[\bGamma_{AB}^{(s)} + {\bf \tilde U_A}\bGamma_{AB}^{(1-s)}{\bf \tilde U_A}^\intercal]}}\label{ris: Qs quantiy}\\
\Delta_s&=\tilde\beeta_U^{(A)\intercal}\left[\bGamma_{AB}^{(s)} + {\bf \tilde U_A}\bGamma_{AB}^{(1-s)}{\bf \tilde U_A}^\intercal\right]^{-1}\tilde\beeta_U^{(A)},\label{eq: Delta expression}
\end{align}
and with 
\begin{eqnarray} \label{DEFTILDE} 
{\bf \tilde U_A} = {\bf U}_A\oplus\bId_B, \qquad 
\tilde\beeta_U^{(A)} = (\beeta_U^{(A)},\bz^{(B)}),
\end{eqnarray}  
being respectively the extensions to $B$ [see  Eq.~(\ref{eq: U_A action})] of the local symplectic matrix ${\bf U}_A$ and of the local vector $\beeta_U^{(A)}$ which define the action of the Gaussian unitary $U_A$ via the correspondence~(\ref{eq: Gauss unitary action}).

Enforcing hence the restriction~(\ref{def: form S lambda}) discussed in Sec.~\ref{sec: Unitary choice}, in constructing our Gaussian version of the DS 
we shall use  ${\bf U}_A$ and   $\beeta_U^{(A)}$  of the form given in Eqs.~(\ref{def: UA rotation1}) and (\ref{def: vec1}),
optimizing the resulting expression  with respect to  ${\bf V}_A$ and  $\beeta_V^{(A)}$ which parametrize the elements of the set $\mathcal{S}_{\{\lambda_{j}\}}$. 
Notice, however, that while the minimization  with respect to  ${\bf V}_A$ can be expected, we can intuitively get rid of the displacement parameter $\beeta_V^{(A)}$.
Indeed, when dealing with non-classicality measures one is always interested in the worst-case choice of local $U_A$, that alters the states as little as possible. For this reason the choice $\beeta_V^{(A)}=\bz$, leading to a null displacement on the state, is expected to be optimal in all situations of interest.
 Exploiting Eqs.~\eqref{eq: Formal QCB} and \eqref{eq: Delta expression}, we are now in the position of showing that at least for the Gaussian version of DS 
 this is indeed the case. 
For this purpose, observe that the matrix $\bGamma_{AB}^{(s)} + {\bf \tilde U_A}\bGamma_{AB}^{(1-s)}{\bf \tilde U_A}^\intercal$ is positive definite, yielding $\Delta_s\geq 0$ with equality for $\beeta_V^{(A)}=\bz$, and hence $Q_s e^{-\Delta_s} \leq Q_s$.
This last inequality holds also taking the minimum over $s$ and the maximum over $U_A$, so that we can write
\begin{equation}
\max_{{\bf V}_A,\beeta_V^{(A)}} Q(\rho,U_A\,\rho\, U_A^\dagger) \leq \max_{{\bf V}_A,\,\beeta_V^{(A)}=\bz}  Q(\rho,U_A\,\rho\, U_A^\dagger).
\end{equation}
Since keeping $\beeta_V^{(A)}$ fixed corresponds to consider a smaller maximization set, we can conclude that we are allowed to drop all displacements, maximizing only over ${\bf V_A} \in Sp(2n_A,\mathbb{R})$. Therefore, without loss of generality, we finally define
 the Gaussian Discriminating Strength  (GDS) as:
\begin{eqnarray}\label{def: GaussianDS}
&&{\cal GDS}_A^{ \{\lambda_{j}\} }(\rho_{AB})  \\ \nonumber 
&& \qquad = 1-\max_{{\bf V_A}}\min_{s\in[0,1]} \frac{\Pi_{j=1}^{n} \left[\Lambda_s(\nu_j) + \Lambda_{1-s}(\nu_j) \right]}{\sqrt{\det[\bGamma_{AB}^{(s)} + {\bf \tilde U_A}\bGamma_{AB}^{(1-s)}{\bf \tilde U_A}^\intercal]}},
\end{eqnarray}
with ${\bf \tilde U_A}$ linked to ${\bf  V_A}$ as detailed previously.

\subsection{Two-mode case}
We will now obtain closed  expressions for the GDS when referred to special classes of two-mode bipartite systems (i.e. $n=2$ and $n_A=n_B=1$). In performing the optimization
of Eq.~(\ref{def: GaussianDS}) we use the Euler decomposition \cite{EulerDec_Book} of a symplectic matrix to parametrize ${\bf V}_A$, according to which every single mode
${\bf S}\in Sp(2)$ can be written as 
 \begin{equation}\label{def: Euler Decomposition}
{\bf S} = {\bf R}(\theta){\bf S}^{(1)}(x){\bf R}(\theta'),
\end{equation}
where ${\bf R}(\theta), {\bf R}(\theta')\in SO(2)$ as in Eq.~(\ref{def: R}), and 
\begin{equation}\label{def: R and S1 matrix}
{\bf S}^{(1)}(x)=\left(\begin{array}{cc}
e^{+x} & 0 \\
0 &	e^{-x}
\end{array}\right),
\end{equation}
with $x$ being a real parameter.
 
Moreover, since ${\bf V}_A$ intervenes in Eq.~(\ref{def: UA rotation1}) always in the product 
${\bf U}_A={\bf V}_A\,{\bf R}_A(\lambda)\,{\bf V}_A^{-1}$, when $n_A=1$ only the squeezing and the rotation on the left are relevant in the Euler decomposition \eqref{def: Euler Decomposition}, being $SO(2)$ abelian. Therefore, ${\bf V_A}$ can be effectively
expressed as 
\begin{equation}\label{def: VA parametrization}
{\bf V}_A (\theta, x)= {\bf R}_A(\theta)\,{\bf S}^{(1)}_A(x),
\end{equation}
yielding from  Eq.~(\ref{def: UA rotation1}) the following functional dependence for the symplectic matrix of $U_A$, 
\begin{eqnarray}\label{explicit UA}
{\bf U}_A(\theta,x) = 
{\bf R}_A(\theta)\,{\bf S}^{(1)}_A(x){\bf R}_A(\lambda){\bf S}^{(1)}_A(-x){\bf R}_A(-\theta), \nonumber \\\label{impoEQ} 
\end{eqnarray}
$\theta\in [0,2\pi[$ and $x \in \; ] \!-\! \infty , \infty[$ being the parameters over which the maximization  of Eq.~(\ref{def: GaussianDS}) has to be taken for fixed $\lambda \neq 2\pi n$, $n\in\mathbb Z$.
The resulting GDS for the two-mode case  becomes  hence  
\begin{eqnarray}\label{def: GaussianDS2modes}
&& {\cal GDS}_A^{(\lambda)}(\rho_{AB})  \\ 
\nonumber && \quad = 1-\max_{\theta, x}\min_{s\in[0,1]} \frac{ \Pi_{j=1}^{2} \left[\Lambda_s(\nu_j) + \Lambda_{1-s}(\nu_j) \right]
}{\sqrt{F^{(\lambda)}_s(\theta,x)}},
\end{eqnarray}
where we introduced the function
\begin{eqnarray} 
&& F_s^{(\lambda)}(\theta, x) 
\\ &&\;\; 
\nonumber = \det[\bGamma_{AB}^{(s)} + ( {\bf U}_A(\theta,x)\oplus\bId_B)\bGamma_{AB}^{(1-s)}({\bf U}^\intercal_A(\theta,x)\oplus\bId_B)] ,
\end{eqnarray} 
which bares the GDS dependence upon $\lambda$ via Eq.~(\ref{impoEQ}).

We also observe that from the invariance  under local Gaussian operations of our functional~[see property $M1$ of Theorem 1], the covariance matrix $\Gamma_{AB}$  of any two-mode input Gaussian state $\rho_{AB}$ can be considered in the standard form \cite{NormalForm1,NormalForm2,Ferraro_Review}: 
\begin{equation}\label{def: Standard Form}
\bGamma_A = a\,\bId_2,\quad \bGamma_B = b\,\bId_2, \quad\bGamma_{OFF} = \text{Diag}(c,d),
\end{equation} in terms of $2\times 2$ blocks defined as in Eq.~\eqref{eq: covariance matrix blocks}.
Unfortunately even with all these simplifications the explicit evaluation of \eqref{def: GaussianDS2modes} on every two-mode Gaussian state is still not trivial. This is because the calculations cannot be carried out directly using the coefficients $a,b,c,d$ of Eq.~(\ref{def: Standard Form}), since the quantity $Q_s$ of Eq.~(\ref{ris: Qs quantiy}) is expressed in terms of the Williamson decomposition \eqref{def: Williamson decomposition}, where an high number of parameters intervene in the parametrization of the symplectic matrix ${\bf S}$ \cite{Serafini_SympInvariants}.
For this reason from now on we will consider only two classes of two-mode Gaussian states, characterized by 
$|c|=|d|$ 
in the standard form. Such states are obtained from thermal states by means of linear mixing  ($c=d$), or two-mode squeezing ($c=-d$). Most importantly, in both cases the symplectic matrix ${\bf S}$ that appears in the Williamson
decomposition~\eqref{def: Williamson decomposition} of the covariance matrix $\Gamma_{AB}$ can be described in terms of a single parameter~\cite{Weedbrook_GaussianReview,Ferraro_Review}.
In particular, in terms of the Pauli matrix $\bsigma_3 = \text{Diag(1,-1)}$, for a two-mode squeezed thermal state $\rho_{AB}^{(\rm sq)}$ one identifies
 the symplectic matrix of~\eqref{def: Williamson decomposition} with 
\begin{equation}\label{def: symp squeezing}
{\bf S}_{\rm sq}(r) = \left(\begin{array}{c|c}
\cosh r \,\bId_2 & \sinh r \,\bsigma_3\\ \hline
\sinh r \,\bsigma_3 & \cosh r \,\bId_2
\end{array}\right),
\end{equation}
 the connection with Eq.~(\ref{def: Standard Form}) being provided by the expressions \cite{Weedbrook_GaussianReview,Adesso_SympInvariants}
\begin{equation}\label{eq: stand parameter minus 1}
\nu_1-\nu_2 = a-b, \quad (\nu_1 + \nu_2)^2 = (a+b)^2 - 4c^2,
\end{equation}  
\begin{equation}\label{eq: stand parameter minus 2}
\sinh^2(2r) = \frac{4c^2}{(a+b)^2 - 4c^2},
\end{equation}
 $\nu_1$ and $\nu_2$ being the symplectic eigenvalues of $\Gamma_{AB}$. 
For a thermal state after a linear mixing via a beam splitter  $\rho_{AB}^{(\rm lm)}$, instead, Eq.~(\ref{def: symp squeezing}) gets replaced by 
\begin{equation}\label{def: symp linear mixing}
{\bf S}_{\rm lm}(\phi) = \left(\begin{array}{c|c}
\cos \phi \,\bId_2 & -\sin \phi \,\bId_2\\ \hline
\sin \phi \,\bId_2 & \cos \phi \,\bId_2
\end{array}\right),
\end{equation}
while Eqs.~(\ref{eq: stand parameter minus 1}) and (\ref{eq: stand parameter minus 2}) by the identities 
\begin{equation}\label{eq: stand parameter plus 1}
\nu_1+\nu_2 = a+b, \quad (\nu_1 - \nu_2)^2 = (a-b)^2 + 4c^2,
\end{equation}  
\begin{equation}\label{eq: stand parameter plus 2}
\sin^2(2\phi) = \frac{4c^2}{(a-b)^2 + 4c^2}.
\end{equation}

\subsubsection{Explicit evaluation on linear mixing and two-mode squeezing of thermal states} 
It is important to stress that the minimum over the parameter $s$ 
in Eq.~\eqref{def: GaussianDS} is in general difficult to evaluate, because it has to be performed for a generic choice of the parameters $\theta$ and $x$ which, via Eq.~(\ref{impoEQ}),
define the symplectic matrices ${\bf U}_A$ over which we have to take the minimization. A lower bound on GDS can be obtained setting $s=1/2$ throughout the computation. In \cite{Illuminati_GaussianDoR} it is shown that when the rotation parameter in \eqref{impoEQ} is set to be $\lambda=\pm\pi/2$, this bound is achieved for every Gaussian state.
In Appendix \ref{app: s minimum} we show that for the two aforementioned classes of states, such minimum is actually reached in $s=1/2$ for every choice of $\lambda$. However, we have numerical evidences that this is no more true for a generic two-mode Gaussian state not fulfilling the symmetric condition~$|c|=|d|$.

Applying the aforementioned result in Appendix \ref{app: s minimum}, we now evaluate \eqref{def: GaussianDS2modes} setting directly $s=1/2$, obtaining 
\begin{eqnarray}\label{def: GaussianDS2modesSIMP}
 {\cal GDS}_A^{(\lambda)}(\rho_{AB})  = 1-\frac{ 4(A_+^2 - A_-^2)}
{\sqrt{\min_{\theta, x}  F_{1/2}^{(\lambda)}(\theta,x)}},
\end{eqnarray}
where in writing the nominator we introduced the quantities 
\begin{equation}\label{def: A plus minus}
	A_{\pm} = \frac{\Lambda_{1/2}(\nu_1) \pm \Lambda_{1/2}(\nu_2)}{2}.
\end{equation}
After long but straightforward calculations, we can also express the term at the denominator as 
\begin{eqnarray}
F_{1/2}^{(\lambda)}(\theta,x) &=& \left[4 (A_+^2-A_-^2)+4 \sin^2(\lambda/2)\; S\;\right]^2 \nonumber \\
&+& 16 \sinh^2(2x) \sin^2(\lambda) (A_+^2-A_-^2)\; C,  \label{defFmezzo} 
\end{eqnarray}
where $S$ and $C$ are  positive quantities  defined as 
\begin{eqnarray} 
S&=&A_{+}^2 \sinh^2(2r), \\  
 C&=&A_{+}^2\cosh^2(2r)-A_{-}^2, 
 \end{eqnarray}
 for the squeezed thermal states $\rho^{(\rm sq)}_{AB}$ of Eq.~(\ref{def: symp squeezing}), and  
 \begin{eqnarray} 
S&=&A_{-}^2 \sin^2(2\phi), \\  
 C&=& A_{+}^2  - A_{-} ^2\cos^2(2\phi), 
 \end{eqnarray} 
 for the thermal states after a linear mixing $\rho^{(\rm lm)}_{AB}$ of Eq.~(\ref{def: symp linear mixing}). 
Notice that in Eq.~(\ref{defFmezzo}) there is no dependence upon $\theta$ and that the left hand side  reaches the minimum for 
$x=0.$ 
Accordingly we can write 
\begin{align}
{\cal GDS}_A^{(\lambda)}\left(\rho^{(\rm sq)}_{AB}\right) &= \frac{\sinh^2(2r)\sin^2(\lambda/2)}{\left[1-\left(\frac{A_-}{A_+}\right)^2\right] + \sinh^2(2r)\sin^2(\lambda/2)},\label{ris: GDS squeezing}\\
{\cal GDS}_A^{(\lambda)}\left(\rho^{(\rm lm)}_{AB}\right) &= \frac{\sin^2(2\phi)\sin^2(\lambda/2)}{\left[\left(\frac{A_+}{A_-}\right)^2 - 1\right] + \sin^2(2\phi)\sin^2(\lambda/2)} \label{ris: GDS lin mix}.
\end{align}
The optimal measure for these classes is therefore the one described by the parameter $\lambda=\pi$, that yields the maximum amount of correlations for every given state (e.g. see Figure \ref{fig: GDS}). Notice also that due to the dependence upon $A_-^2$, the obtained result is invariant under the exchange of the two subsystems [see Eq.~\eqref{def: A plus minus}]. This feature appeared also in the Gaussian Interferometric Power \cite{Adesso_GaussianIP}, and it is a peculiarity of the considered classes of states. 
Qualitatively we can see that fixing the squeezing $r$ or the linear mixing parameter $\phi$, GDS approaches its maximum value $1$ when $A^2_-\simeq A^2_+$, that from \eqref{def: A plus minus} corresponds of having a big gap between the symplectic eigenvalues $\nu_1$ and $\nu_2$ of $\bGamma_{AB}$ (i.e. in the excitation numbers of the thermal states). This is analogous to the result pointed out in \cite{Illuminati_GaussianDoR} for the Gaussian Discord of Response, where the amount of correlations can be increased with the asymmetry on the number of thermal photons between the two subsystems. Notice also that symmetric states ($\nu_1=\nu_2$) yield the same amount of non-classicality as pure states (which is zero for the linear mixing case), and that for two-mode squeezed states the maximum value of $1$ can be obtained for every thermal state in the limit of $r \to \infty$. Similarly, the maximum amount of correlations obtained through linear mixing is achieved by means of a balanced beam splitter ($\phi = \pi/4$), but it can approach $1$ only for highly asymmetric initial thermal states. For more informations on the correlating power of beam splitters, see \cite{Paris_BSCorrelations,Kim_BSCorrelations1,Kim_BSCorrelations2}.

\begin{figure}
	\centering
	\includegraphics[width=0.45\textwidth]{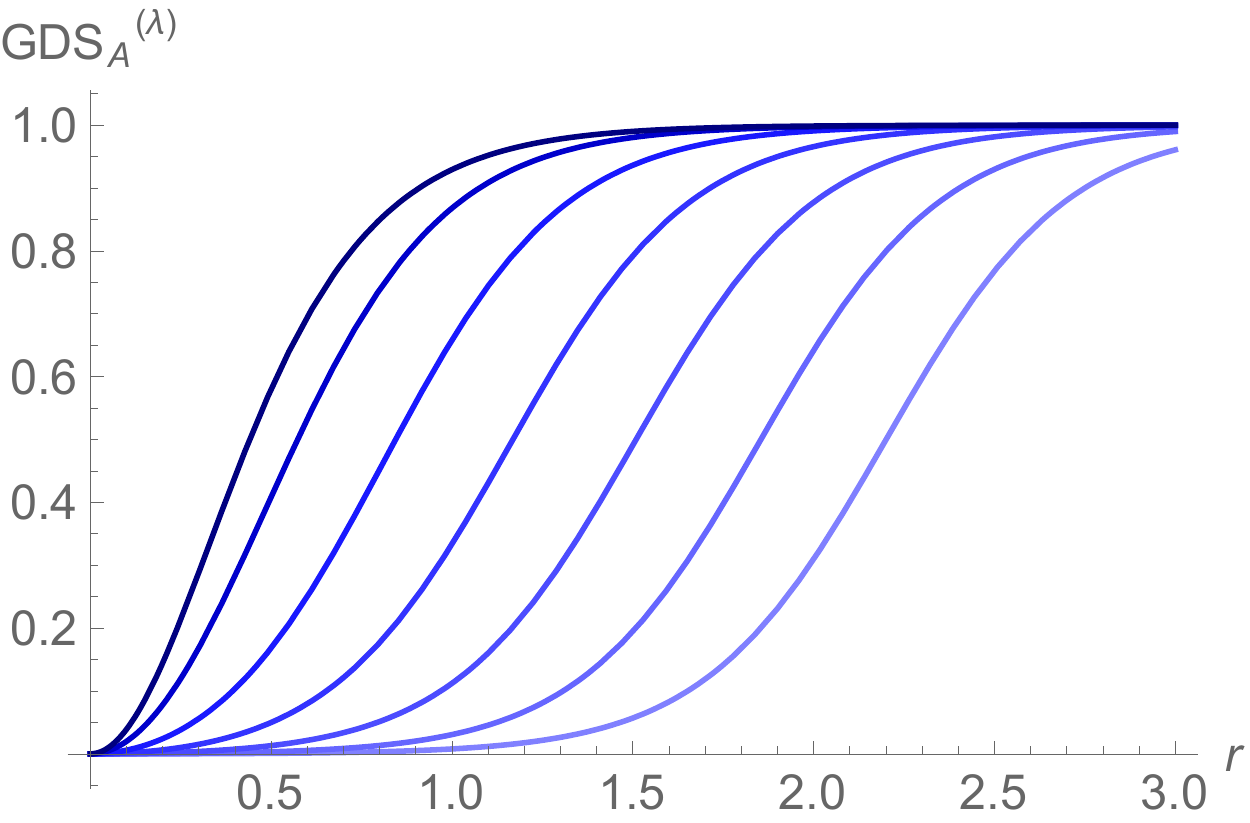}
	\caption{(Color online) GDS of the two-mode squeezed vacuum state as a function of the squeezing parameter. In the plot we have used a set of exponentially decreasing parameters $\lambda=\frac{\pi}{2^k}$. From left to right we find the curves for $k = 0, 1,\dots 6$, the optimal choice being associated with $k=0$, i.e. $\lambda = \pi$.  \label{fig: GDS}}
\end{figure}

The expressions obtained here are written in terms of the parameters that intervene in the Williamson decomposition \eqref{def: Williamson decomposition}: exploiting Eqs.~(\ref{eq: stand parameter minus 1}),~(\ref{eq: stand parameter minus 2}), (\ref{eq: stand parameter plus 1}) and (\ref{eq: stand parameter plus 2}) we can, however, convert them in terms of the 
the parameters $a$, $b$, $c$ and $d$ which appears in the  standard form \eqref{def: Standard Form}.
In particular 
for symmetric two-mode squeezed thermal states, i.e. $a=b$ or $\nu_1=\nu_2$, this yields 
\begin{equation} \label{symmm} 
{\cal GDS}_A^{(\lambda)}(\rho_{AB}^{(\rm sq,sym)}) = \frac{c^2 \sin^2(\lambda/2)}{a^2-c^2 \cos^2(\lambda/2)},
\end{equation}
that when $\lambda=\pi/2$ coincides with the related result on the Quantum Chernoff Bound obtained in \cite{Illuminati_GaussianDoR}. For asymmetric states, however, the contribution of $(A_-/A_+)^2$ becomes cumbersome when expressed as a function of standard form parameters. 
An  expression analogous to \eqref{symmm}  can be derived also for the linear mixing case.
Indeed for  $a=b$ (which in this case does not correspond to have $\nu_1=\nu_2$) the GDS reads:
\begin{eqnarray}
&& {\cal GDS}_A^{(\lambda)}(\rho_{AB}^{(\rm lm,sym)}) \\ &&\quad =\frac{4c^2\sin^2(\lambda/2)}{\left[\sqrt{(a+c)^2-1} + \sqrt{(a-c)^2-1}\right]^2- 4c^2\cos^2(\lambda/2)}. \nonumber 
\end{eqnarray}

\subsection{Relation between GDS and entanglement - total number of photons}\label{sec: GDS relation}
In the previous sub-section we showed that keeping the squeezing or the mixing parameter fixed, the amount of correlations increases with the asymmetry in the number of thermal photons used. Here we want to discuss how the GDS behaves for the same two classes of states with respect to entanglement or to the total number of photons. In other words, our goal here is the identification of the optimal state achieving the maximum amount of correlations, when we fix either the degree of entanglement or the total energy contained in the state.

We begin by considering entanglement on the squeezed class only, being $\rho_{AB}^{(\rm lm)}$ always separable. To quantify it, we can use the logarithmic negativity \cite{Vidal_LogNeg,Plenio_LogNeg}, defined as
\begin{equation}\label{def: log neg}
\mathcal{E} = \max\{-\log(\tilde\nu_-),0\},
\end{equation}
where $\tilde{\nu}_-$ is the smallest symplectic eigenvalue of the partially transposed state. Exploiting the method based on the symplectic invariants \cite{Serafini_SympInvariants,Adesso_SympInvariants}, $\tilde \nu_-$ can be explicitly written as:
\begin{equation}\label{eq: PPT NuMinus}
	\mathcal{\tilde\nu_-} = \sqrt{\frac{\tilde{\Delta} - \sqrt{\tilde{\Delta }^2 - 4 \det\bGamma_{AB}}}{2}},
\end{equation}
where $\tilde{\Delta} = \det\bGamma_A + \det\bGamma_B -2\det\bGamma_{OFF}$. For a two-mode squeezed thermal state this corresponds to $\det\bGamma_{AB} = \nu_1^2\nu_2^2$ and
\begin{equation}\label{def: delta tilde}
\tilde{\Delta} = \cosh(4r)\left(\frac{\nu_1+\nu_2}{2}\right)^2 + \left(\frac{\nu_1-\nu_2}{2}\right)^2,
\end{equation}
so that an explicit expression for $\mathcal E (\nu_1,\nu_2,r)$ can be obtained (see Eq.~\eqref{app: explicit f} in Appendix \ref{app: ent-N relation} if interested). In particular, for the pure two-mode squeezed vacuum one has $\mathcal E(1,1,r) = 2|r|$. Therefore, for this class of pure states GDS is a monotonic function of the logarithmic negativity. At least in this particular case, we can see how property $3.$ of a good measure of non-classical correlations (see Sec. \ref{sec: Intro}), namely of being an entanglement monotone on pure states, is satisfied.

We numerically evaluated  GDS (with the optimal value $\lambda=\pi$) and the logarithmic negativity for a set of $10^6$ two-mode squeezed thermal states randomly generated with the constraints $1\leq\nu_1=\nu_2\leq 20$ and $0\leq r \leq 5$. The result is shown in Figure \ref{fig: GE}, where we put in evidence the behavior of symmetric states with $\nu_1=\nu_2$ in both the pure (red lower line) and the mixed case (black dashed lines). Notice that for every fixed value of entanglement GDS is minimized on the pure two-mode squeezed vacuum, on which the entanglement is easily related to the squeezing parameter by the relation $\mathcal E(1,1,r) = 2|r|$. Such bound holds even if we allow $\nu_1$ to be different from $\nu_2$.
On the other side, considering mixed states the value of GDS is allowed to go as close to $1$ as desired. These facts are formally stated in the following proposition, which is proven in Appendix \ref{app: ent-N relation}. 

\begin{prop}\label{thm: ent relation}
	Consider ${\cal GDS}_A^{(\lambda)}(\nu_1,\nu_2,r)$ for a two-mode squeezed thermal state, and its logarithmic negativity $\mathcal E(\nu_1,\nu_2,r)$. Then:
	\begin{itemize}
		\item Fixed $\mathcal{E}$, the GDS is minimized on pure states:
		\begin{align}
		{\cal GDS}_A^{(\lambda)}(\nu_1,\nu_2,r)\geq  {\cal GDS}_A^{(\lambda)}\left(1,1,\mathcal E(\nu_1,\nu_2,r)/2\right);\notag
		\end{align}
		\item For all $\mathcal E\geq0$ and $\Delta$ such that
		\begin{equation}
		{\cal GDS}_A^{(\lambda)}(1,1,\mathcal E /2)\leq \Delta< 1,
		\end{equation}
		there exists a symmetric state with $\nu_1=\nu_2=\nu$ with entanglement $\mathcal E$ and ${\cal GDS}_A^{(\lambda)} = \Delta$.
	\end{itemize}
\end{prop}
This means that if entanglement is the quantity in which we are interested in optimizing, then pure states offer the worst performance, whereas mixed states with the same entanglement yield always a greater GDS. 

\begin{figure}
	\centering
	\includegraphics[width=0.45\textwidth]{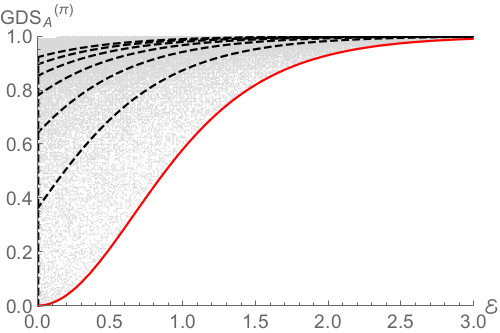}
	\caption{(Color online) GDS with respect to logarithmic negativity $\mathcal E$ for a set of $10^6$ randomly generated two-mode squeezed thermal states with $1\leq \nu_1=\nu_2=\nu \leq 20$ and $0\leq r \leq 5$. The curves corresponding to pure states (red, lower bound) and to symmetric states with $\nu = 2,3,4,5,6,7$ (black dashed lines from bottom to top) are also shown. The plot is realized with $\lambda=\pi$. \label{fig: GE}}
\end{figure}

The situation, however, changes if we take into account also the energy that must be used to produce such thermal excitations. This can be done comparing GDS with the total number of photons $N$ of the state. For a linear mixed thermal state this reads:
\begin{align}
N_{\rm lm} &= \Tr\left[(\aop^\dagger\aop + \bop^\dagger\bop)S_{\rm lm}(\phi)\rho_{\rm th}S^\dagger_{\rm lm}(\phi)\right]\notag\\
&=\Tr\left[(\aop^\dagger\aop + \bop^\dagger\bop)\rho_{\rm th}\right] = \frac{\nu_1+\nu_2}{2} - 1,
\end{align}
where we used the fact that a beam splitter does not change the total number of photons. In the previous expression $\rho_{\rm th}$ is the thermal state with symplectic eigenvalues $\nu_1,\nu_2$, while $S_{\rm lm}(\phi) = e^{-\phi\left(\aop^\dagger\bop -\bop^\dagger\aop\right)}$ is the unitary operation that corresponds to the symplectic matrix ${\bf S}_{\rm lm}(\phi)$ of Eq.~\eqref{def: symp linear mixing} through relations \eqref{eq: Gauss unitary action}. In this case the total number of photons fixes exactly the sum $\nu_1 + \nu_2$, therefore the maximum of GDS \eqref{ris: GDS lin mix} is obtained when the gap between them is as big as possible: $\nu_1 = 2N_{lm} + 1$ and $\nu_2=1$ or viceversa. From this it follows that the maximum amount of correlations that a balanced beam splitter can generate, starting from thermal states with a total number of photons $N$, is given by:
\begin{equation}
GDS_A^{(\lambda)}(\rho_{AB}^{(\rm lm, opt)}) = \frac{N-N\cos\lambda}{2+N-N\cos\lambda}.
\end{equation}
With the optimal choice $\lambda=\pi$, it reduces to the simple expression
\begin{equation}
	GDS_A^{(\pi)}(\rho_{AB}^{(\rm lm, opt)}) = \frac{N}{N+1},
\end{equation}
that is plotted as a blue dot-dashed line in Figure \ref{fig: GN}.

\begin{figure}
	\centering
	\includegraphics[width=0.45\textwidth]{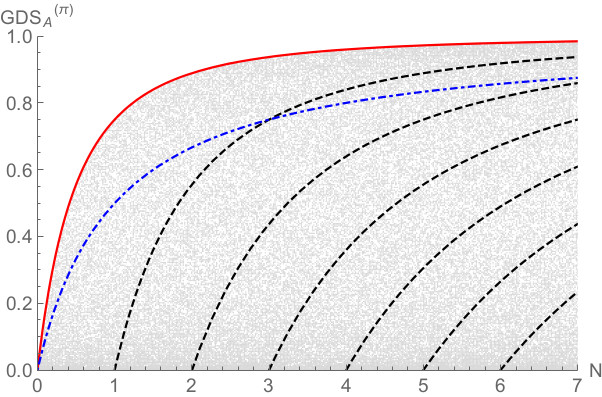}
	\caption{(Color online) Gaussian DS with respect to the total number of photons for a set of $10^6$ randomly generated two-mode squeezed thermal states with $1\leq \nu_1=\nu_2=\nu \leq 8$ and $0\leq r \leq 7$. The curves corresponding to pure states (red, upper bound) and to symmetric states with $\nu = 2,3,4,5,6,7$ (black dashed lines from left to right) are also shown. The blue dot-dashed line that crosses the others corresponds to the optimal GDS achieved only by means of linear mixing of thermal states. The plot is realized with $\lambda=\pi$.\label{fig: GN}}
\end{figure}

In order to do an equivalent calculation for a two-mode squeezed state characterized by the symplectic matrix \eqref{def: symp squeezing}, we need to consider a thermal state evolved with the operator $S_{\rm sq}(r)=e^{r(\aop^\dagger\bop^\dagger-\aop\bop)}$,
that modifies the total number of photons as:
\begin{equation}
S_{\rm sq}^\dagger (\aop^\dagger\aop + \bop^\dagger\bop)S_{\rm sq} = 2\sinh^2 r +\cosh(2r)(\aop^\dagger\aop + \bop^\dagger\bop).
\end{equation}
Therefore in this case the total number of photons depends also upon the squeezing:
\begin{align}
N_{\rm sq} &= \Tr\left[(\aop^\dagger\aop + \bop^\dagger\bop)S_{\rm sq}(r)\rho_{\rm th}S^\dagger_{\rm sq}(r)\right]\notag\\
&=2\sinh^2 r + \cosh(2r)\Tr\left[(\aop^\dagger\aop + \bop^\dagger\bop)\rho_{\rm th}\right] \notag\\
&= \cosh(2r)\frac{\nu_1+\nu_2}{2} - 1,\label{eq: N for squeezed states}
\end{align}
For the case of pure states this expression reduces to $N_{sq}(1,1,r)=2\sinh^2 r$.
Analogously to what have been previously done for entanglement, 
we can numerically study the relation between GDS (with $\lambda=\pi$) and the total number of photons for a set of $10^6$ randomly generated Gaussian states with $1\leq \nu_1=\nu_2 \leq 8$ and $0\leq r \leq 7$. The result is plotted in Figure \ref{fig: GN}, where the behaviour of pure states (red upper line) and symmetric mixed states with $\nu_1=\nu_2$ (black dashed lines) is evidenced. Notice that, differently than before, pure states corresponds to the GDS upper bound for of a two-mode squeezed thermal state, when the number photons is kept fixed. As in the entanglement case, such bound holds also for $\nu_1\neq\nu_2$. Every other value of correlations below the threshold of pure states can instead be achieved. Both these facts are formally stated by the following proposition, that is proven in Appendix \ref{app: ent-N relation}.

\begin{prop}\label{thm: photon number relation}
		Consider ${\cal GDS}_A^{(\lambda)}(\nu_1,\nu_2,r)$ for a two-mode squeezed thermal state, and its total number of photons $N_{sq}(\nu_1,\nu_2,r)$. Then:
		\begin{itemize}
			\item Fixed $N_{sq}$, ${\cal GDS}_A^{(\lambda)}$ is maximized on pure states:
			\begin{align}
 {\cal GDS}_A^{(\lambda)}&(\nu_1,\nu_2,r)\notag \leq\\ & GDS_A^{(\lambda)}\left(1,1,\text{arcsinh}\left(\sqrt{\frac{N_{sq}(\nu_1,\nu_2,r)}{2}}\right)\right);\notag
			\end{align}
			\item For all $N_{sq}\geq0$ and $\Delta$ such that \begin{equation}
			0\leq \Delta\leq {\cal GDS}_A^{(\lambda)}\left(1,1,\text{arcsinh}\left(\sqrt{N_{sq}/2}\right)\right),
			\end{equation}
			there exists a symmetric state with $\nu_A=\nu_B=\nu$ with total number of photons $N_{sq}$ and ${\cal GDS}_A^{(\lambda)} = \Delta$.
		\end{itemize}
\end{prop}

Comparing Proposition \ref{thm: ent relation} with Proposition \ref{thm: photon number relation} we can see that, depending on which resource we take into account, pure states can be either the optimal or the worst possible choice. They indeed allow to reach the maximum amount of correlations when the total number of photons is fixed. However, if thermal excitations are already present in the system, maybe due to some noise, then the mixedness that is introduced increases the non-classical behaviour of the state when the same squeezing parameter or the same amount of entanglement is considered.

Eventually, in Figure \ref{fig: GN} one can see also a comparison between the maximum GDS that can be activated from thermal states, by means of a balanced beam splitter, or via a two-mode squeezing operation leading to the same total number of photons. It can be appreciated how the two-mode squeezed vacuum always performs better, but at the cost of utilizing highly-squeezed light, which is much harder to produce with respect to the combination of thermal light and linear mixing through a balanced beam splitter.

\section{Conclusions} \label{sec: Conclusions}
The susceptibility of a state to the less disturbing local unitary map has been widely used in the literature to quantify the amount of non-classical correlations. In order to obtain a bona-fide measure, which nullifies only on the set of CQ states, such unitary is usually chosen among those operations characterized by a fixed non-degenerate spectrum. In the first part of this paper we studied how such condition is modified when a Gaussian restriction is imposed on both the state and the unitary maps involved in the minimization. In particular, in Sec. \ref{sec: Unitary choice} we showed that, under few very reasonable assumptions, all local unitary operations must be chosen among those that can be obtained from a non-trivial operator of the form $e^{i\sum_{j=1}^{n_A}\lambda_j \aop_j^\dagger\aop_j}$ by means of local Gaussian unitary maps. Although its single-mode version has been assumed several times in the literature, to the best of our knowledge no explicit proof of its necessity was known so far.  
Moreover, we discussed how usually the optimization over the remaining local Gaussian unitary map can be performed only by considering its action on the covariance matrix, without taking into account the possibility of a displacement.

After this general discussion, we focused on a particular measure of non-classical correlations: the Discriminating Strength. With a clear operational meaning, it aims at enlightening the usefulness of correlations in a state discrimination setup. We obtained an expression for its Gaussian version (GDS), defined considering Gaussian input states and restricting the optimization set $\mathcal S$ to a subset of all Gaussianity-preserving unitary maps, as described in the first part of this manuscript. We then computed the obtained quantity on two simple, but still relevant, classes of two-mode states: those obtained from thermal states by means of a linear mixing or a two-mode squeezing operation. Explicit expressions for  GDS have been provided in these cases, in terms of the parameters intervening in their Williamson decomposition, as well as in terms of their standard form defined up to local unitary maps. We showed that given an initial thermal state, the maximum amount of non-classical correlations that can be obtained by mixing linearly the two modes is achieved with a balanced beam splitter, whereas increasing the squeezing parameter always leads to a bigger amount of non-classicality. Furthermore, asymmetric states with a different number of thermal excitations in the two modes, i.e. with symplectic eigenvalues $\nu_1\neq\nu_2$, are always more correlated than their symmetric counterpart.

We compared GDS in a two-mode squeezed state with the amount of entanglement, as measured by the logarithmic negativity. We showed that, for every fixed value of entanglement, pure states offer the worst performance in terms of amount of correlations, whereas mixed states can lead to values of GDS as close to the maximum bound of $1$ as desired. However, this is due to the fact that states with initial strong thermal populations can be highly squeezed despite their small amount of entanglement. Therefore, it is also interesting to compare different states at fixed energy. From the analysis of the relation between GDS and photon number emerged that when the total energy is kept fixed the pure two-mode squeezed vacuum is indeed the option leading to the biggest amount of correlations. Eventually, we showed that mixing different thermal states in a balanced beam splitter still allows the generation of non-classical correlations. Despite the fact that this scenario does not achieve values as high as a two-mode squeezed vacuum, it has the advantage of not requiring highly-squeezed light, which is harder to obtain experimentally with respect to many-photons thermal states.

\section*{Acknowledgements}
We warmly thank G. Adesso for useful comments and discussions. LR acknowledges financial support from the People Programme (Marie Curie Actions) of the European Union's Seventh Framework Programme (FP7/2007-2013) under REA grant agreement n$\degree$ 317232; ADP and VG are supported by the EU Collaborative Project TherMiQ (Grant agreement 618074) and by the EU project COST Action MP1209 ``Thermodynamics in the quantum regime''.

\bibliographystyle{unsrt} 
\bibliography{Bibliography_GDS} 

\begin{thebibliography}{10}

\bibitem{Horodecki_EntReview}
R.~Horodecki, P.~Horodecki, M.~Horodecki, and K.~Horodecki.
\newblock Quantum entanglement.
\newblock {\em Rev. Mod. Phys.}, 81:865--942, Jun 2009.

\bibitem{Modi_Review}
K.~Modi, A.~Brodutch, H.~Cable, T.~Paterek, and V.~Vedral.
\newblock The classical-quantum boundary for correlations: Discord and related
  measures.
\newblock {\em Rev. Mod. Phys.}, 84:1655--1707, Nov 2012.

\bibitem{Zurek_QDiscord}
H.~Ollivier and W.~H. Zurek.
\newblock Quantum discord: A measure of the quantumness of correlations.
\newblock {\em Phys. Rev. Lett.}, 88:017901, Dec 2001.

\bibitem{Vedral_QDiscord}
L.~Henderson and V.~Vedral.
\newblock Classical, quantum and total correlations.
\newblock {\em Journal of Physics A: Mathematical and General}, 34(35):6899,
  2001.

\bibitem{Horodecki_QDeficit}
M.~Horodecki, P.~Horodecki, R.~Horodecki, J.~Oppenheim, A.~Sen, U.~Sen, and
  B.~Synak-Radtke.
\newblock Local versus nonlocal information in quantum-information theory:
  Formalism and phenomena.
\newblock {\em Phys. Rev. A}, 71:062307, Jun 2005.

\bibitem{Modi_RelEntropy}
K.~Modi, T.~Paterek, W.~Son, V.~Vedral, and M.~Williamson.
\newblock Unified view of quantum and classical correlations.
\newblock {\em Phys. Rev. Lett.}, 104:080501, Feb 2010.

\bibitem{Orszag_BuresGeometric}
D.~Spehner and M.~Orszag.
\newblock Geometric quantum discord with bures distance.
\newblock {\em New Journal of Physics}, 15(10):103001, 2013.

\bibitem{Orszag_BuresGQubit}
D.~Spehner and M.~Orszag.
\newblock Geometric quantum discord with bures distance: the qubit case.
\newblock {\em Journal of Physics A: Mathematical and Theoretical},
  47(3):035302, 2014.

\bibitem{Sarandy_Geom1Norm}
F.~M. Paula, T.~R. de~Oliveira, and M.~S. Sarandy.
\newblock Geometric quantum discord through the schatten 1-norm.
\newblock {\em Phys. Rev. A}, 87:064101, Jun 2013.

\bibitem{Giovannetti_TDDQubit}
F.~Ciccarello, T.~Tufarelli, and V.~Giovannetti.
\newblock Toward computability of trace distance discord.
\newblock {\em New Journal of Physics}, 16(1):013038, 2014.

\bibitem{Adesso_NoQ}
T.~Nakano, M.~Piani, and G.~Adesso.
\newblock Negativity of quantumness and its interpretations.
\newblock {\em Phys. Rev. A}, 88:012117, Jul 2013.

\bibitem{Piani_GDiscordProblem}
M.~Piani.
\newblock Problem with geometric discord.
\newblock {\em Phys. Rev. A}, 86:034101, Sep 2012.

\bibitem{Luo_Measurement}
Shunlong Luo and Shuangshuang Fu.
\newblock Measurement-induced nonlocality.
\newblock {\em Phys. Rev. Lett.}, 106:120401, Mar 2011.

\bibitem{Luo_MID}
S.~Luo.
\newblock Using measurement-induced disturbance to characterize correlations as
  classical or quantum.
\newblock {\em Phys. Rev. A}, 77:022301, Feb 2008.

\bibitem{Adesso_GaussMID}
L.~Mi\ifmmode~\check{s}\else \v{s}\fi{}ta, R.~Tatham, D.~Girolami,
  N.~Korolkova, and G.~Adesso.
\newblock Measurement-induced disturbances and nonclassical correlations of
  gaussian states.
\newblock {\em Phys. Rev. A}, 83:042325, Apr 2011.

\bibitem{Fu_NonLocality}
L.~B. Fu.
\newblock Nonlocal effect of a bipartite system induced by local cyclic
  operation.
\newblock {\em EPL (Europhysics Letters)}, 75(1):1, 2006.

\bibitem{Datta_Inequivalence}
A.~Datta and S.~Gharibian.
\newblock Signatures of nonclassicality in mixed-state quantum computation.
\newblock {\em Phys. Rev. A}, 79:042325, Apr 2009.

\bibitem{Girolami_IP}
D.~Girolami, A.~M. Souza, V.~Giovannetti, T.~Tufarelli, J.~G. Filgueiras, R.~S.
  Sarthour, D.~O. Soares-Pinto, I.~S. Oliveira, and G.~Adesso.
\newblock Quantum discord determines the interferometric power of quantum
  states.
\newblock {\em Phys. Rev. Lett.}, 112:210401, May 2014.

\bibitem{Paris_QFI}
M.~G.~A. Paris.
\newblock Quantum estimation for quantum technology.
\newblock {\em International Journal of Quantum Information},
  07(supp01):125--137, 2009.

\bibitem{DS}
A.~Farace, A.~De Pasquale, L.~Rigovacca, and V.~Giovannetti.
\newblock Discriminating {S}trength: a bona fide measure of non-classical
  correlations.
\newblock {\em New Journal of Physics}, 16(7):073010, 2014.

\bibitem{Audenaert_QCB}
K.~M.~R. Audenaert, J.~Calsamiglia, R.~Mu\~noz Tapia, E.~Bagan, Ll. Masanes,
  A.~Acin, and F.~Verstraete.
\newblock Discriminating states: The quantum chernoff bound.
\newblock {\em Phys. Rev. Lett.}, 98:160501, Apr 2007.

\bibitem{Illuminati_DoR}
W.~Roga, S.~M. Giampaolo, and F.~Illuminati.
\newblock Discord of response.
\newblock {\em Journal of Physics A: Mathematical and Theoretical},
  47(36):365301, 2014.

\bibitem{Ferraro_Review}
A.~Ferraro, S.~Olivares, and {M. G. A.} Paris.
\newblock {\em Gaussian states in continuous variable quantum information}.
\newblock Napoli Series on Physics and Astrophysics (ed. Bibliopolis, Napoli),
  2005.

\bibitem{Weedbrook_GaussianReview}
C.~Weedbrook, S.~Pirandola, R.~Garc\'ia-Patr\'on, N.~J. Cerf, T.~C. Ralph,
  J.~H. Shapiro, and S.~Lloyd.
\newblock Gaussian quantum information.
\newblock {\em Rev. Mod. Phys.}, 84:621--669, May 2012.

\bibitem{Wang_Review}
X.~B. Wang, T.~Hiroshima, A.~Tomita, and M.~Hayashi.
\newblock Quantum information with gaussian states.
\newblock {\em Physics Reports}, 448(1–4):1 -- 111, 2007.

\bibitem{Paris_GaussianDiscord}
P.~Giorda and M.~G.~A. Paris.
\newblock Gaussian quantum discord.
\newblock {\em Phys. Rev. Lett.}, 105:020503, Jul 2010.

\bibitem{Adesso_GaussianDiscord}
G.~Adesso and D.~Girolami.
\newblock Gaussian geometric discord.
\newblock {\em International Journal of Quantum Information},
  09(07n08):1773--1786, 2011.

\bibitem{Adesso_GaussianDiscordOld}
G.~Adesso and A.~Datta.
\newblock Quantum versus classical correlations in gaussian states.
\newblock {\em Phys. Rev. Lett.}, 105:030501, Jul 2010.

\bibitem{Giovannetti_GaussConjecture1}
A.~Mari, V.~Giovannetti, and A.~S. Holevo.
\newblock Quantum state majorization at the output of bosonic gaussian
  channels.
\newblock {\em Nature Communications}, 5:3826, May 2014.

\bibitem{Giovannetti_GaussConjecture2}
V.~Giovannetti, R.~Garc\'{\i}a-Patr\'{o}n, N.~J. Cerf, and A.~S. Holevo.
\newblock Ultimate classical communication rates of quantum optical channels.
\newblock {\em Nature Photonics}, 8:796–800, Sept 2014.

\bibitem{Pirandola_OptimalityGaussDiscord}
S.~Pirandola, G.~Spedalieri, S.~L. Braunstein, N.~J. Cerf, and S.~Lloyd.
\newblock Optimality of gaussian discord.
\newblock {\em Phys. Rev. Lett.}, 113:140405, Oct 2014.

\bibitem{Adesso_GaussianIP}
G.~Adesso.
\newblock Gaussian interferometric power.
\newblock {\em Phys. Rev. A}, 90:022321, Aug 2014.

\bibitem{Illuminati_GaussianDoR}
W.~Roga, D.~Buono, and F.~Illuminati.
\newblock Device-independent quantum reading and noise-assisted quantum
  transmitters.
\newblock {\em New Journal of Physics}, 17(1):013031, 2015.

\bibitem{Williamson}
J.~Williamson.
\newblock On the algebraic problem concerning the normal forms of linear
  dynamical systems.
\newblock {\em American Journal of Mathematics}, 58(1):141--163, 1936.

\bibitem{Rahimi-Keshari_DiscordVerification}
S.~Rahimi-Keshari, C.~M. Caves, and T.~C. Ralph.
\newblock Measurement-based method for verifying quantum discord.
\newblock {\em Phys. Rev. A}, 87:012119, Jan 2013.

\bibitem{Adesso_CQGaussian}
L.~Mi\ifmmode~\check{s}\else \v{s}\fi{}ta, D.~McNulty, and G.~Adesso.
\newblock No-activation theorem for gaussian nonclassical correlations by
  gaussian operations.
\newblock {\em Phys. Rev. A}, 90:022328, Aug 2014.

\bibitem{Pirandola_LambdaGFunctions}
S.~Pirandola and S.~Lloyd.
\newblock Computable bounds for the discrimination of gaussian states.
\newblock {\em Phys. Rev. A}, 78:012331, Jul 2008.

\bibitem{Marian_Fidelity}
P.~Marian and T.~A. Marian.
\newblock Uhlmann fidelity between two-mode gaussian states.
\newblock {\em Phys. Rev. A}, 86:022340, Aug 2012.

\bibitem{EulerDec_Book}
Arvind, B~Dutta, N~Mukunda, and R~Simon.
\newblock The real symplectic groups in quantum mechanics and optics.
\newblock {\em Pramana}, 45(6):471--497, 1995.

\bibitem{NormalForm1}
L.M. Duan, G.~Giedke, J.~I. Cirac, and P.~Zoller.
\newblock Inseparability criterion for continuous variable systems.
\newblock {\em Phys. Rev. Lett.}, 84:2722--2725, Mar 2000.

\bibitem{NormalForm2}
R.~Simon.
\newblock Peres-horodecki separability criterion for continuous variable
  systems.
\newblock {\em Phys. Rev. Lett.}, 84:2726--2729, Mar 2000.

\bibitem{Serafini_SympInvariants}
A.~Serafini, F.~Illuminati, and S.~De Siena.
\newblock Symplectic invariants, entropic measures and correlations of gaussian
  states.
\newblock {\em Journal of Physics B: Atomic, Molecular and Optical Physics},
  37(2):L21, 2004.

\bibitem{Adesso_SympInvariants}
G.~Adesso, A.~Serafini, and F.~Illuminati.
\newblock Extremal entanglement and mixedness in continuous variable systems.
\newblock {\em Phys. Rev. A}, 70:022318, Aug 2004.

\bibitem{Paris_BSCorrelations}
M.~Brunelli, C.~Benedetti, S.~Olivares, A.~Ferraro, and M.~G.~A. Paris.
\newblock Single- and two-mode quantumness at a beam splitter.
\newblock {\em Phys. Rev. A}, 91:062315, Jun 2015.

\bibitem{Kim_BSCorrelations1}
M.~S. Kim, W.~Son, V.~Bu\ifmmode~\check{z}\else \v{z}\fi{}ek, and P.~L. Knight.
\newblock Entanglement by a beam splitter: Nonclassicality as a prerequisite
  for entanglement.
\newblock {\em Phys. Rev. A}, 65:032323, Feb 2002.

\bibitem{Kim_BSCorrelations2}
S.~C. Springer, J.~Lee, M.~Bellini, and M.~S. Kim.
\newblock Conditions for factorizable output from a beam splitter.
\newblock {\em Phys. Rev. A}, 79:062303, Jun 2009.

\bibitem{Vidal_LogNeg}
G.~Vidal and R.~F. Werner.
\newblock Computable measure of entanglement.
\newblock {\em Phys. Rev. A}, 65:032314, Feb 2002.

\bibitem{Plenio_LogNeg}
M.~B. Plenio.
\newblock Logarithmic negativity: A full entanglement monotone that is not
  convex.
\newblock {\em Phys. Rev. Lett.}, 95:090503, Aug 2005.

\bibitem{Book_Det_Estimation}
C.~H. Helstrom.
\newblock {\em Quantum Detection and Estimation Theory}.
\newblock Academic Press (February 11, 1976).

\bibitem{Lloyd_Illumination}
S.~Lloyd.
\newblock Enhanced sensitivity of photodetection via quantum illumination.
\newblock {\em Science}, 321(5895):1463--1465, 2008.

\bibitem{Giovannetti_GaussianIllumination}
S.~H. Tan, B.~I. Erkmen, V.~Giovannetti, S.~Guha, S.~Lloyd, L.~Maccone,
  S.~Pirandola, and J.~H. Shapiro.
\newblock Quantum illumination with gaussian states.
\newblock {\em Phys. Rev. Lett.}, 101:253601, Dec 2008.

\bibitem{Lloyd_Illumination2}
J.~H. Shapiro and S.~Lloyd.
\newblock Quantum illumination versus coherent-state target detection.
\newblock {\em New Journal of Physics}, 11(6):063045, 2009.

\bibitem{Guha_GaussianIllumination}
S.~Guha and B.~I. Erkmen.
\newblock Gaussian-state quantum-illumination receivers for target detection.
\newblock {\em Phys. Rev. A}, 80:052310, Nov 2009.

\bibitem{Dariano_ICPOVM1}
G.~M. D’Ariano, P.~Perinotti, and M.~F. Sacchi.
\newblock Informationally complete measurements and group representation.
\newblock {\em Journal of Optics B: Quantum and Semiclassical Optics},
  6(6):S487, 2004.

\bibitem{Scott_IPOVM2}
A.~J. Scott.
\newblock Tight informationally complete quantum measurements.
\newblock {\em Journal of Physics A: Mathematical and General}, 39(43):13507,
  2006.

\bibitem{PhysRevA.66.032316}
G\'eza G. and J.~I.~Cirac.
\newblock Characterization of gaussian operations and distillation of gaussian
  states.
\newblock {\em Phys. Rev. A}, 66:032316, Sep 2002.

\bibitem{Book_DeGosson}
M.~A. De~Gosson.
\newblock {\em Symplectic Methods in Harmonic Analysis and in Mathematical
  Physics}.
\newblock Basel: Birkh\"{a}user, 2011.

\end{thebibliography}

\newpage
\appendix

 \section{State-discrimination protocol defining the DS}\label{app: DS protocl}
 In this section we review the state-discrimination protocol that intervenes in the definition of the measure of non-classical correlations known as DS introduced in~\cite{DS}. Before going into details, it is worth to recall the Quantum Chernoff Bound (QCB) \cite{Audenaert_QCB}, that describes the probability of error when discriminating between two known quantum states $\rho_0$ and $\rho_1$, available in $M$ copies. Optimizing over all Positive Operator Valued Measurements (POVM), if there is a $50 \%$ a priori probability of having $\rho_0^{\otimes M}$ or $\rho_1^{\otimes M}$ one has \cite{Book_Det_Estimation}:
 \begin{equation}
 P_{err,min}^{(M)}(\rho_0,\rho_1)=\frac{1}{2}\left(1-||\rho_0^{\otimes M} - \rho_1^{\otimes M}||_1\right),
 \end{equation} 
 corresponding to an optimal strategy that discriminates between the positive and negative eigenspaces of the difference $\rho_0^{\otimes M} - \rho_1^{\otimes M}$. In the limit $M \gg 1$ such probability scales exponentially in the number of copies:
 \begin{equation}
 P_{err,min}^{(M)}(\rho_0,\rho_1)\simeq e^{-M \xi(\rho_0,\rho_1)},
 \end{equation}
 where we defined
 \begin{equation}
 \xi(\rho_0,\rho_1) = -\lim_{M\to \infty}\frac{\log P_{err,min}^{(M)}(\rho_0,\rho_1)}{M}.
 \end{equation}
 Such limit intervenes in the definition of the QCB \cite{Audenaert_QCB}, that appears in Sec. \ref{sec: Gaussian DS} of the main text:
 \begin{equation}\label{def: appQCB}
 e^{-\xi(\rho_0,\rho_1)} = Q(\rho_0,\rho_1) = \min_{s\in[0,1]}\Tr\left[\rho_0^s\,\rho_1^{1-s}\right].
 \end{equation}
 Notice that in particular the QCB can be bounded as:
 \begin{equation}
 0\leq Q(\rho_0,\rho_1) \leq \Tr\left[\rho_0^{1/2}\,\rho_1^{1/2}\right] \leq 1.
 \end{equation}
 
 The DS is based on a quantum illumination scenario \cite{Lloyd_Illumination,Giovannetti_GaussianIllumination,Lloyd_Illumination2,Guha_GaussianIllumination}, where one is interested in obtaining informations about an environment by probing it with a bipartite state: the actual probe $A$ and a reference $B$. In particular the DS considers a situation in which one is interested in detecting the possible application of a local unitary operation performed by a third party, that is only partially cooperative. The protocol can be thought as follows (see Figure \ref{fig: DS scheme}): Alice and Bob choose the initial probing state $\rho_{AB}$, and produce $M$ copies of it. Alice's local state is then sent to a third party Charlie, who can apply or not (with $50 \%$ probability) a unitary operation $U_A$ of his choice. Charlie then communicates $U_A$ (without saying whether he applied it or not) and gives back to Alice the states, on which she can perform a generic joint POVM on the whole system $AB$ with the help of Bob. Their aim is to discriminate between $\rho_{AB}$ and $U_A\rho_{AB}U_A^\dagger$, in order to detect whether the local rotation actually took place or not. 
 
 \begin{figure}
 	\centering
 	\includegraphics[width=0.45\textwidth]{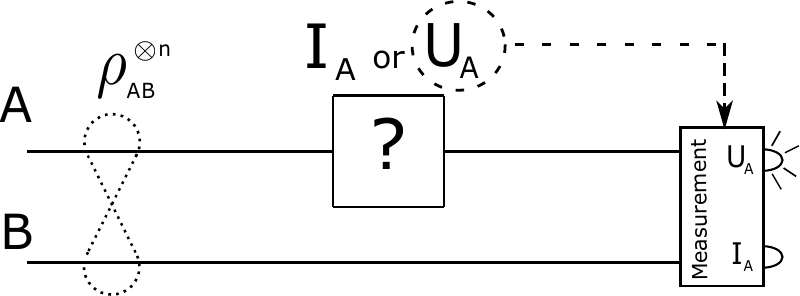}
 	\caption{Scheme for the state discrimination protocol on which the Discriminating Strength is defined. Alice and Bob prepare $n$ copies of a state, without knowing which local unitary operation $U_A$ will be applied in the following. Alice's subsystems are then modified locally by a third party (Charlie), that can choose with the same probability to: $(i)$ leave them unchanged or $(ii)$ to apply a local rotation $U_A\in \mathcal S$ of his choice. The picked unitary is then communicated to Alice and Bob, that can use this information to perform the optimal joint measurement to detect whether Charlie chose option $(i)$ or option $(ii)$. \label{fig: DS scheme}}
 \end{figure}
 
 Notice that the chosen $U_A$, not known at the state preparation stage, is eventually revealed in order to allow Alice and Bob to perform the best possible measurement, whose outcome can be described by the QCB. This guarantees that the worst-case success probability only depends on the initial chosen state $\rho_{AB}$ and on the set $\mathcal S$ of the allowed unitary operations $U_A$.
 As can be expected, the presence of non-classical correlations in the initial state $\rho_{AB}$ turns out to be useful. Consider for example the structure of a CQ (uncorrelated) state \eqref{def: CQ state}: indeed, if Charlie chooses a unitary operation diagonal in the same basis $\{\ket{i}_A\}_i$ in which the CQ state is expanded, then the whole state $\rho_{AB}$ is left unmodified and Alice and Bob have no way to detect the rotation. The more the initial state $\rho_{AB}$ is susceptible to local alterations, the higher the probability of success of Alice and Bob will be.
 Therefore, the Discriminating Strength of the initial state $\rho_{AB}$ defined in Eq. \eqref{def: DS} of the main text quantifies their worst performance in such discrimination task.

 In a previous work \cite{DS} it was shown how the DS is a bona-fide measure of non-classical correlations, satisfying the four properties stated in Sec. \ref{sec: Intro}, when $\mathcal S$ is chosen to be the set of all local unitary operations with a fixed and non-degenerate spectrum.

\section{CQ Gaussian states are completely uncorrelated} \label{app: CQ gaussian states} 
 
In this appendix we rederive that only completely uncorrelated Gaussian states, i.e. decomposed as $\rho_{A}\otimes\rho_B$, are CQ (i.e. with a zero Quantum Discord, see Sec. \ref{sec: Intro}), result that has been firstly obtained in \cite{Adesso_CQGaussian}. The discussion will follow the ideas of \cite{Rahimi-Keshari_DiscordVerification}, where only the two-mode case is considered.

The core of the proof relies on a result involving informationally-complete POVM (IC-POVM) \cite{Dariano_ICPOVM1,Scott_IPOVM2}. A POVM $\{M_k\}_k$, with $\sum_k M_k = \Id$, is said to be complete if it allows to reconstruct the state, i.e. if there exist operators $N_k$ such that $\sigma = \sum_k N_k \Tr[\sigma M_k]$, for every density matrix $\sigma$ of the system being measured.
If we perform an IC-POVM on the subsystem $B$ of a bipartite state $\rho_{AB}$, we are left with the conditional states:
\begin{equation}
\rho_{A|k} = \frac{\Tr_B[\rho_{AB} M_k]}{p_k}, \qquad p_k = \Tr[\rho_{AB} M_k].
\end{equation}
Rahimi-Keshari et al. in \cite{Rahimi-Keshari_DiscordVerification} have proven the following statement.
\begin{approp}\label{RK prop}
	The necessary and sufficient condition for a bipartite system $\rho_{AB}$ to be CQ is that, for every IC-POVM $\{M_k\}_k$ performed on subsystem $B$, the conditional states $\{\rho_{A|k}\}_k$ commute with one another, i.e.\begin{equation}
	[\rho_{A|k},\rho_{A|k^\prime}] = 0, \quad\forall k,k^\prime.
	\end{equation}
\end{approp}
Based on this result, in order to prove that a CQ state $\rho \in \mathfrak G$ has to be completely uncorrelated, we can just prove that if its covariance matrix has some correlations $\bGamma_{OFF} \neq \bz$ [see \eqref{eq: covariance matrix blocks}] then there exist an IC-POVM leading to two conditional states $\rho_{A|k}$, $\rho_{A|k^\prime}$ that do not commute.
In the following we will consider a particular IC-POVM for continuous variable systems: the multi-mode heterodyne detection, that projects on coherent states
\begin{equation}
M_{\bbeta} = \frac{1}{\pi^{n_B}}\ketbras{\bbeta}{\bbeta}{B},
\end{equation}
with $\bbeta\in\mathbb{C}^{n_B}$. Notice that the completeness follows from the correspondence between quantum states and quasi-probability distributions on phase space, that can be reconstructed from such measurement.

To describe the conditional state $\rho_{A|\bbeta}$ after a heterodyne measurement on subsystem $B$, we will use the following well known result on Gaussian POVMs \cite{PhysRevA.66.032316,Weedbrook_GaussianReview}:
\begin{approp}
	Let us consider a bipartite Gaussian state $\rho_{AB}$ described by a covariance matrix and a displacement vector given by:
	\begin{equation}
	\bGamma_{AB} = \left(\begin{array}{c|c}
	\bGamma_A & \bGamma_{OFF}\\\hline
	\bGamma_{OFF}^\intercal & \bGamma_B
	\end{array}\right), 
	\qquad \bxi = \left(\begin{array}{c}
	\bxi_A\\
	\bxi_B
	\end{array}\right),
	\end{equation}
	whose subsystem $B$ is projected on a Gaussian state described by $\bGamma_m$ and $\bxi_m$. After the measurement, the state left in $A$ has covariance matrix and displacement vector given by:
	\begin{align}
	\bGamma_{A|m} &= \bGamma_A - \bGamma_{OFF}\frac{1}{\bGamma_B + \bGamma_m}\bGamma_{OFF}^\intercal, \\
	\bxi_{A|m} &= \bxi_A - \bGamma_{OFF}\frac{1}{\bGamma_B + \bGamma_m}\left(\bxi_B - \bxi_m\right).
	\end{align}
\end{approp}
In our case, up to local Gaussian operators (that do not alter the amount of correlations), we can always choose $\bxi_A = \bxi_B = \bz$. Moreover, in the heterodyne detection $\bGamma_m = \bId_B$ and $\bxi_m = \bbeta \in \mathbb{C}^{n_B}$. The conditional state $\rho_{A|\bbeta}$ will be therefore characterized by
\begin{align}
\bGamma_{A|\bbeta} &= \bGamma_A - \bGamma_{OFF}\frac{1}{\bGamma_B + \bId_B}\bGamma_{OFF}^\intercal,\\
\bxi_{A|\bbeta} &= \bGamma_{OFF}\frac{1}{\bGamma_B + \bId_B}\,\bbeta.
\end{align}
Notice that being $\bGamma_{A|\bbeta}$ independent on $\bbeta$, all the conditional states $\rho_{A|\bbeta}$ differ only for their first moments. Since $\bGamma_{OFF} \neq \bz$ by hypothesis and being $\bGamma_B + \bId_B \in GL(2 n_B)$, there exists a vector $\bbeta_0$ such that $\bxi_{A|\bbeta_0} \neq 0$, whereas choosing $\bbeta = \bz$ one has $\bxi_{A|\bz} = 0$. In the following we will focus on the couple of states $\rho_{A|\bz}$ and $\rho_{A|\bbeta_0}$, proving that they do not commute. 

To do so, we will use the fact that the product $\rho_1\rho_2$ of two states $\rho_1,\rho_2 \in \mathfrak G$ still has a Gaussian characteristic function (possibly with non-real covariance matrix $\bGamma_{12}$) \cite{Marian_Fidelity}:
\begin{equation}\label{ch funxtion}
\chi_{\rho_1\rho_2}(\bu) = \chi(\bz)e^{-\frac{1}{4}\bu^\intercal \bGamma_{12} \bu -i\bxi_{12}^\intercal \bu}.
\end{equation}
In particular, $\bxi_{12}$ can be written as \cite{Marian_Fidelity}:
\begin{equation}\label{product displacement}
\bxi_{12} = \bxi_1 - (\bGamma_1 - i\bOmega)\frac{1}{\bGamma_1+\bGamma_2}(\bxi_1 - \bxi_2),
\end{equation}
where $\bxi_i$ and $\bGamma_i$, $i = 1,2$, are respectively the displacements and covariance matrices of $\rho_1$, $\rho_2$ and $\bOmega$ is the symplectic form \eqref{def: symplectic form}.
Condition $[\rho_1,\rho_2] = 0$ implies $ \bxi_{12} = \bxi_{21}$, because to different characteristic functions correspond different quantum states. From \eqref{product displacement} and its counterpart for $\bxi_{21}$, such equality is equivalent to $\bxi_1 = \bxi_2$. 

In conclusion, given $\rho_{AB} \in \mathfrak G$ with $\bGamma_{OFF} \neq \bz$, we found a continuous variable IC-POVM (heterodyne detection) yielding two conditional states $\rho_{A|\bbeta_0}$ and $\rho_{A|\bz}$ characterized by different displacements $\bxi_{A|\bbeta_0}\neq \bxi_{A|\bz}$. Since this implies their non-commutativity, from Proposition \eqref{RK prop} it follows that $\rho_{AB}$ cannot be CQ. Conversely, if a Gaussian state $\rho_{AB} = \rho_{A} \otimes \rho_{B}$ is completely uncorrelated, it can be explicitly written in a CQ form \eqref{def: CQ state} by diagonalizing its local components $\rho_A$ and $\rho_B$.
 
\section{Standard form for a symplectic orthogonal matrix} \label{app: symp decomposition}
It is a well known fact of linear algebra that a special orthogonal matrix $W\in SO(2n)$ can be reduced to a direct sum of $n$ two-dimensional rotations \eqref{def: R and S1 matrix} by means of an orthogonal change of basis:
\begin{equation}\label{appeq: Omega standard form for W}
{\bf W} = {\bf M}\left(\bigoplus_{j=1}^{n}{\bf R}(\lambda_j)\right) {\bf M}^{\intercal}, \qquad {\bf M}\in O(2n).
\end{equation}
In this appendix we will prove that if ${\bf W}$ is also symplectic, then the same decomposition can be obtained with ${\bf M}\in Sp(2n,\mathbb{R})\cap O(2n)$.
From now on it is convenient to adopt a different quadrature operator vector $\rop$, in which the position-like operators come before the momentum-like ones. With this definition all the relevant properties hold with respect to the symplectic form
\begin{equation}
{\bf J} = \left(\begin{array}{c|c}
 & \bId_n\\\hline
 -\bId_n &
\end{array}\right),
\end{equation}
with zeros in the diagonal blocks.
Notice that with this different ordering the standard form of ${\bf W}$, equivalent to \eqref{appeq: Omega standard form for W}, will be
\begin{equation}\label{appeq: rotation standard form}
{\bf W}_{\oplus}=\left(\begin{array}{c|c}
\bigoplus_j \cos\lambda_j & \bigoplus_j -\sin\lambda_j\\\hline
\bigoplus_j \sin\lambda_j & \bigoplus_j \cos\lambda_j
\end{array}\right).
\end{equation}

In order to complete our proof, we will need two general properties of the symplectic group \cite{Book_DeGosson}, that are stated in the following lemma. The proof is sketched for completeness. Note that a complex matrix $M_{c}$ is said to be symplectic  if it satisfies the same formal requirement as in the real case:
\begin{equation}
{\bf M}_{c} \,{\bf J}\, {\bf M}_{c}^\intercal = {\bf J},
\end{equation}
and we will write ${\bf M}_{c}\in Sp(2n,\mathbb C)$.

\begin{aplem}\label{lemma}
Let us consider the matrix 
\begin{equation}
{\bf K}=\frac{1}{\sqrt{2}}\left(\begin{array}{c|c}
\bId_n & i\bId_n\\\hline
i\bId_n & \bId_n
\end{array}\right),
\end{equation}
that given a real matrix ${\bf M}$ allows for the definition of its complex correspondent as
\begin{equation}
{\bf M}_c = {\bf K}\,{\bf M} \,{\bf K}^{-1}.
\end{equation}
Then:
\begin{itemize}
	\item[$(i)$] ${\bf M}\in Sp(2n,\mathbb R) \;\Longleftrightarrow \;{\bf M}_c\in Sp(2n,\mathbb C)$; 
	\item[$(ii)$] ${\bf M}\in Sp(2n,\mathbb{R})\cap O(2n)\; \Longleftrightarrow\; {\bf M}_c = {\bf U}\oplus{\bf U}^*$, \\for some ${\bf U} \in U(n)$.
\end{itemize}
\end{aplem}
\begin{proof}
$(i)$ The proof of the first statement follows trivially from the symmetry of ${\bf K}$ and the observation 
 	\begin{equation}
 	{\bf K}\,{\bf J}\,{\bf K}={\bf K}^{-1}\,{\bf J}\,{\bf K}^{-1}={\bf J}.
 	\end{equation}

$(ii)$ Let us divide the unitary matrix ${\bf U}$ into its real and imaginary part as
\begin{equation}
{\bf U} = {\bf A} + i {\bf B},
\end{equation}
where the unitarity enforces the symmetry of ${\bf A}{\bf B}^\intercal$ and the relation ${\bf A}{\bf A}^\intercal + {\bf B}{\bf B}^\intercal = \bId$. 
It is straightforward to check that the requirement ${\bf M}_c = {\bf U}\oplus{\bf U}^*$ with ${\bf U} \in U(n)$ is equivalent to
\begin{equation}\label{appeq: M decomposition}
{\bf M} = \left(\begin{array}{c|c}
{\bf A} & -{\bf B} \\ \hline
{\bf B} & {\bf A}
\end{array}\right).
\end{equation}
 We will now prove the equivalence of this last formulation with ${\bf M}\in Sp(2n,\mathbb{R})\cap O(2n)$. If decomposition \eqref{appeq: M decomposition} is assumed, the orthogonality and the symplecticity of the matrix can be explicitly checked showing that indeed ${\bf M}{\bf M}^\intercal = \bId$ and ${\bf M}{\bf J}{\bf M}^\intercal={\bf J}$. Viceversa, if ${\bf M}\in Sp(2n,\mathbb{R})\cap O(2n)$, then
\begin{equation}
{\bf J}\,{\bf M}={\bf M}^{-1\intercal}\,{\bf J} = {\bf M}\,{\bf J},
\end{equation}
which yields the block structure of \eqref{appeq: M decomposition}. The required relations among the blocks follow trivially from the orthogonality of the matrix.
\end{proof}

\begin{thm}
	Every symplectic orthogonal matrix ${\bf W} \in Sp(2n,\mathbb{R})\cap O(2n)$ can be transformed in its standard form ${\bf W}_{\oplus}$ by means of ${\bf M}\in Sp(2n,\mathbb{R})\cap O(2n)$:
	\begin{equation}
	{\bf M}^\intercal\, {\bf W} \,{\bf M} = {\bf W}_{\oplus}.
	\end{equation}
\end{thm}
\begin{proof}
	From Lemma \ref{lemma} it follows that ${\bf W}_c$ can be written in terms of a unitary matrix ${\bf U} \in U(n)$, that can be diagonalized with a unitary matrix ${\bf O} \in U(n)$. Setting ${\bf M}_c = {\bf O}\oplus{\bf O}^*$ one has:
	\begin{equation}\label{appeq: Wc expression}
	{\bf W}_c = {\bf U}\oplus {\bf U}^* = {\bf M}_c\left(\begin{array}{c|c}
	\bigoplus_j e^{i\lambda_j} & \\\hline
	& \bigoplus_j e^{-i\lambda_j}
	\end{array}\right) {\bf M}_c^\dagger,
	\end{equation}
	with $e^{i\lambda_j}$ being the eigenvalues of ${\bf U}$. 
	Being ${\bf W}_c = {\bf K}\,{\bf W}\,{\bf K}^{-1}$ by definition, and noticing that
	\begin{equation}
	{\bf W}_\oplus = {\bf K}^{-1}\left(\begin{array}{c|c}
	\bigoplus_j e^{i\lambda_j} & \\\hline
	& \bigoplus_j e^{-i\lambda_j}
	\end{array}\right){\bf K},
	\end{equation}
	setting ${\bf M} = {\bf K}^{-1}\,{\bf M}_c \,{\bf K} $ Eq.~\eqref{appeq: Wc expression} yields
	\begin{equation}
 {\bf W}  = {\bf M}\,{\bf W}_{\oplus}\,	{\bf M}^\intercal.
	\end{equation}
	The definition of ${\bf M}_c$ and Lemma \ref{lemma} eventually assure that ${\bf M}\in Sp(2n,\mathbb{R})\cap O(2n)$.
\end{proof}

\section{Minimum in s = 1/2 for two-mode squeezing and linear mixing of thermal Gaussian states} \label{app: s minimum}
In this appendix we want to show that the minimum in $s$, that appears in the definition of the Quantum Chernoff Bound \eqref{def: QCB}, is reached for $s=1/2$ when the two classes of states described in the main text are considered, no matter the parameter $\lambda$ used in the optimization set $\mathcal S$. Similarly to the proof in \cite{Illuminati_GaussianDoR} for the $\lambda=\pi/2$ case, the idea is based on the convexity in $s$ of the Quantum Chernoff Bound \cite{Audenaert_QCB}. Instead of showing the symmetry of the involved quantity 
\begin{equation}\label{app: QCB}
\Tr\left[\rho^s \sigma^{1-s}\right], \quad \sigma = U_A \,\rho\, U_A^\dagger,
\end{equation}
under the exchange $s \leftrightarrow 1-s$, we will show that its derivative is zero on $s=1/2$ when the considered state is chosen among those obtained from thermal states via linear mixing or two-mode squeezing. This will yield a statement that holds for every $\lambda$.

Using the cyclicity of the trace the derivative of \eqref{app: QCB} can be written as 
	\begin{equation}
	\frac{\mathrm{d}}{\mathrm{d}s} \Tr\left[\rho^s \sigma^{1-s}\right]=\Tr\left[\left(\log\rho - \log\sigma\right)\rho^s\sigma^{1-s}\right],
	\end{equation}
where up to local unitary maps we can always assume $\rho$ as having null first moments. The same hypothesis can be made on $\sigma$, since the optimization in $U_A$ can always be performed among Gaussian operations purely quadratic in $\rop$ (see Sec. \ref{sec: Gaussian DS} for the proof). 
Being defined as Gibbs state of an Hamiltonian at most quadratic in $\rop$, an undisplaced Gaussian state can always be written as
\begin{equation}
\rho=e^{-\frac{1}{2}\rop^\intercal {\bf H_\rho} \rop},
\end{equation}
where ${\bf H_\rho}$ is diagonalized into some ${\bf D}_\rho$ in terms of the same ${\bf S} \in Sp(2n)$ that intervenes in the Williamson decomposition \eqref{def: Williamson decomposition}:
\begin{equation}\label{app: Gauss on H}
{\bf H_\rho} = {\bf S}^{-1\intercal} {\bf D}_\rho {\bf S}^{-1}.
\end{equation}
Therefore also the Gaussian unitary operation $U_A$, that transform the covariance matrix as in \eqref{eq: Gaussian action on moments}, alters the Hamiltonian in a similar way:
	\begin{equation}
	\log\rho - \log \sigma = - {\bf H}_\rho +  {\bf H}_\sigma =\frac{1}{2}\rop^\intercal \left({\bf \tilde U}_A^{-1\intercal}{\bf H}_\rho {\bf \tilde U}_A^{-1} -  {\bf H}_\rho\right)\rop.
	\end{equation}

Up to a normalization factor proportional to $\Tr(\rho^s)\Tr(\sigma^{1-s})$, the derivative in which we are interested in reads:
	\begin{equation}
\sum_{ij}\left({\bf \tilde U}_A^{-1\intercal}{\bf H}_\rho {\bf \tilde U}_A^{-1} -  {\bf H}_\rho\right)_{ij}\Tr\left[\rho(s)\sigma(1-s)\{r_i,r_j\}_+\right],
\end{equation}
where $\rho(s)$, $\sigma(s)$ are the Gaussian states proportional to $\rho^s$, $\sigma^{1-s}$ and $\{\cdot,\cdot\}_+$ is the anti-commutator. In \cite{Marian_Fidelity} is shown that the product of two Gaussian states $\rho_1\rho_2$ has a characteristic function \eqref{ch funxtion} described by a matrix
\begin{equation}
	\bGamma_{12} = -i\bOmega + \left(\bGamma_2 + i\bOmega\right)\left(\bGamma_1+\bGamma_2\right)^{-1}\left(\bGamma_1 + i\bOmega\right) \label{app: eq Gamma pr},
\end{equation}
where $\bGamma_1,\bGamma_2$ are the covariance matrices of $\rho_1$, $\rho_2$ and $\bOmega$ is the symplectic form \eqref{def: symplectic form}. From definition \eqref{def: Gauss moments2}, and being 
\begin{equation}
	\Tr[\rho^s\sigma^{1-s}] = \Tr[\rho^\frac{s}{2}\sigma^{1-s}\rho^\frac{s}{2}]\in\mathbb R,
\end{equation}
it follows that, up to an irrelevant factor, the quantity that has to be zero for $s=1/2$ is the real part of
\begin{equation}\label{appF: trace}
	\Tr\left[\left({\bf \tilde U}_A^{-1\intercal}{\bf H}_\rho {\bf \tilde U}_A^{-1} -  {\bf H}_\rho\right)\bGamma_{12}\right],
	\end{equation}
where we need to substitute $\bGamma_2={\bf \tilde U}_A\bGamma^{(1-s)}{\bf \tilde U}_A^\intercal$ and $\bGamma_1=\bGamma^{(s)}$ in expression \eqref{app: eq Gamma pr} for $\bGamma_{12}$ [as in the main text, here $\bGamma^{(s)}$ is the covariance matrix of $\rho(s)$, obtained substituting its symplectic eigenvalues $\{\nu_j\}_j$ with $\{\Lambda_s(\nu_j)\}_j$ \eqref{def: lambda function}].
	
	 Exploiting \eqref{app: Gauss on H}, and labelling as ${\bf D}_s$ the block diagonal matrix defining the thermal part of $\bGamma^{(s)}$, the whole real part of \eqref{appF: trace} can be rearranged as:
\begin{widetext}
	\begin{equation}\label{appF: final trace}
	\Tr\left[{\bf D}_\rho \left(\frac{1}{{\bf O}^\intercal {\bf D}_s^{-1} {\bf O} +  {\bf D}_{1-s}^{-1} }-\frac{1}{{\bf O}^{-1\intercal} {\bf D}_{1-s}^{-1} {\bf O}^{-1} + {\bf D}_s^{-1}}\right) + {\bf D}_\rho \left(\frac{1}{{\bf O}^\intercal {\bf D}_s {\bf O} +  {\bf D}_{1-s} }-\frac{1}{{\bf O}^{-1\intercal} {\bf D}_{1-s} {\bf O}^{-1} + {\bf D}_s}\right)\right],
	\end{equation}
\end{widetext}
where the matrix ${\bf O}$ is defined as ${\bf S}^{-1}{\bf \tilde U}_A{\bf S}$.
Notice that to prove that such expression is zero when $s=1/2$, it is enough to show that the quantity
\begin{equation}
\Tr\left[{\bf D}_\rho \left(\frac{1}{{\bf O}^\intercal {\bf D}_{1/2} {\bf O} +  {\bf D}_{1/2} }\right)\right]
\end{equation} 
is invariant under the exchange ${\bf O} \leftrightarrow {\bf O}^{-1}$ for every block diagonal matrix ${\bf D}_\rho$ and ${\bf D}_{1/2}$ (whose blocks are multiples of the identity). Equivalently, from definition~(\ref{DEFTILDE}) and ${\bf U}_A$ explicit expression \eqref{explicit UA}, this quantity must be an even function of $\lambda$ when the symplectic matrix ${\bf S}$ is chosen to be ${\bf S}_{sq}(r)$ \eqref{def: symp squeezing} or ${\bf S}_{lm}(\phi)$ \eqref{def: symp linear mixing}. This can be checked by explicit calculation with the help of a software that allows for analytic manipulations, since the high number of parameters involved makes it not trivial to obtain by hands. 

\section{Relation with Entanglement and total number of photons} \label{app: ent-N relation}
In this appendix we will prove the two propositions stated in sub-Sec. \ref{sec: GDS relation}, concerning the relation between GDS for a two-mode squeezed thermal state and other quantities: the logarithmic negativity and the total number of photons. 

\begin{proof}[Proof of Proposition \ref{thm: ent relation}]
	Following the structure of the statement, we divide the proof in two parts: first we show that pure states give the minimum GDS allowed for a generic state when the entanglement is fixed, and then we prove that the whole region above such value is actually admissible.
 \begin{itemize}
 	\item We need to prove the first statement, which is:
 	\begin{align}
 	{\cal GDS}_A^{(\lambda)}(\nu_1,\nu_2,r)\geq  {\cal GDS}_A^{(\lambda)}\left(1,1,\mathcal E(\nu_1,\nu_2,r)/2\right).
 	\end{align}
 	This can be done in three steps. At first notice that the Gaussian DS \eqref{ris: GDS squeezing} reaches its minimum for a given $r$ when $A_- = 0$, i.e. when $\nu_1=\nu_2$. In the particular case of pure states $(\nu_1=\nu_2=1)$, this yields:
 	\begin{equation}
 	{\cal GDS}_A^{(\lambda)}(\nu_1,\nu_2,r)\geq {\cal GDS}_A^{(\lambda)}(1,1,r).
 	\end{equation} 
 	Secondly, on a pure state GDS is a monotonically increasing function with respect to the absolute value of the squeezing parameter $r$. This can be seen deriving ${\cal GDS}_A^{(\lambda)}(1,1,r)$:
 	\begin{equation}
 	\frac{\mathrm{d}}{\mathrm{d}r} {\cal GDS}_A^{(\lambda)}(1,1,r)=\frac{\sinh(2r)/\sin^2(\lambda/2)}{\left[\cosh^2r+\cot^2(\lambda/2)\right]^2},
 	\end{equation}
 	quantity of that has the same sign of $r$.
 	Therefore in order to conclude this part of the proof, we just need to show the following inequality:
 	\begin{equation}\label{appeq: first thesis}
 	2 |r| = \mathcal E (1,1,r)\geq \mathcal E (\nu_1,\nu_2,r).
 	\end{equation}
 	From definitions \eqref{def: log neg}, \eqref{eq: PPT NuMinus} and \eqref{def: delta tilde} given in the main text, an explicit expression for $\mathcal E$ can be derived: $\mathcal{E} = \max\{0,f(\nu_1,\nu_2,r)\}$, where
 	\begin{widetext}
 		\begin{equation}\label{app: explicit f}
 		f(\nu_1,\nu_2,r)=-\frac{1}{2}\log\left[\cosh(4r)x_+^2 + x_-^2 -\sqrt{\left(\cosh(4r)x_+^2 + x_-^2\right)^2-\nu_1^2\nu_2^2}\right], \quad \text{ 		with}\quad	x_\pm=\frac{1}{2}\left(\nu_1\pm\nu_2\right).
 		\end{equation}
 		\end{widetext}
 		
 	If $f(\nu_1,\nu_2,r)<0$ in \eqref{app: explicit f}, the thesis is trivial. We need to prove that the result holds even when ${ f>0 }$. To do so let us suppose without loss of generality that $r\geq 0$. After some manipulations \eqref{appeq: first thesis} can be shown to be equivalent to the inequality:
 	\begin{align}\label{appeq: long ineq}
 	e^{4r}&\left[(2\nu_1\nu_2)^2-(\nu_1+\nu_2)^2\right]\geq\notag \\
 	& e^{-4r}\left[(\nu_1+\nu_2)^2-4\right]+2(\nu_1-\nu_2)^2,
 	\end{align}
 	where both sides are positive. Since the left-hand side increases with $r$, while the right-hand one decreases with it, we just need to prove the inequality for $r=0$, which is the worst-case scenario. In this case \eqref{appeq: long ineq} can be rearranged as
 	\begin{equation}\label{eq: final ineq}
 	1 + \nu_1^2\nu_2^2 \geq \nu_1^2+\nu_2^2,
 	\end{equation}
 	which holds when $\nu_1,\nu_2\geq 1$.
 \item The second point of the proposition states that two-mode squeezed thermal states cover the entire area above the threshold of pure states and below the upper bound ${\cal GDS}^{(\lambda)}_A = 1$, in the ${\cal GDS}^{(\lambda)}_A - \mathcal{E}$ plot (see for example Figure \ref{fig: GE}). Instead of fixing $\mathcal{E}$, varying $\Delta$ on a $\mathcal{E}$-dependent region, it is easier to do the opposite. Moreover, we can consider only symmetric states with $\nu_1=\nu_2=\nu$. Their GDS does not depend on the actual value of $\nu$, and it is a monotonic increasing function of the squeezing parameter $|r|$, with values in $[0,1[$ (see Figure \ref{fig: GDS} or Eq.~\eqref{ris: GDS squeezing} with $A_-=0$). The logarithmic negativity on the other side has a value of $2|r|-\log\nu$ when $2|r|\geq \log\nu$ and zero otherwise. Therefore, the parameter $0\leq \Delta<1$, defining the value of GDS, fixes the squeezing $|r|$, while $\mathcal{E}$ determines the needed $\nu\geq 1$, that can be obtained only if $\mathcal E$ is not too big with respect to the chosen $|r|$ (or equivalently $\Delta$). It is easy to see that such extreme value for $\mathcal E$ exactly coincides with the one associated with pure two-mode squeezed states.
 \end{itemize}
\end{proof}

\begin{proof}[Proof of Proposition \ref{thm: photon number relation}]
	As for the previous proposition, we divide the proof in two parts: first we prove that pure states achieve the maximum GDS for two-mode squeezed thermal states when the total number of photon is fixed, and then we show that the whole area below such value (and above ${\cal GDS}^{(\lambda)}_A=0$) is admissible.
	
\begin{itemize}
	\item Let us begin by proving
		\begin{align}\label{appeq: thesis 2}
		{\cal GDS}_A&^{(\lambda)}(\nu_1,\nu_2,r)\notag \leq\\ & {\cal GDS}_A^{(\lambda)}\left(1,1,\text{arcsinh}\left(\sqrt{\frac{N_{sq}(\nu_1,\nu_2,r)}{2}}\right)\right).
		\end{align}
	Using Eq.~\eqref{ris: GDS squeezing} for ${\cal GDS}_A^{(\lambda)}$ in the considered class of states, exploiting relation 
	\begin{equation}
	\sinh^2\left[\text{arcsinh}\sqrt{\frac{N_{sq}}{2}}\right] = N_{sq}(N_{sq}+2),
	\end{equation}
	as well as expression \eqref{eq: N for squeezed states} for $N_{sq}(\nu_1,\nu_2,r)$, one can show that \eqref{appeq: thesis 2} is equivalent to
	\begin{widetext}
	\begin{equation}\label{appeq: long ineq 2}
	\cosh^2(2r)\leq\cosh^2(2r)\left(\frac{\nu_1+\nu_2}{2}\right)^2\left[1-\left(\frac{A_-}{A_+}\right)^2\right]+\left(\frac{A_-}{A_+}\right)^2.
	\end{equation}
\end{widetext}
We are now going to prove the stronger inequality obtained from \eqref{appeq: long ineq 2} by dropping the last squared term on the right, so that the dependence upon $r$ disappears. Using definitions \eqref{def: A plus minus} for $A_\pm$, after some manipulations the new inequality can be written as:
\begin{equation}
1-\nu_1\nu_2 \leq \sqrt{\nu_1^2-1}\sqrt{\nu_2^2-1},
\end{equation}
which is evidently true being $\nu_1,\nu_2\geq 1$.

\item  Let us consider the region of the plot ${\cal GDS}_A^{(\lambda)}-N_{sq}$ delimited by the boundary of the pure states and the axis ${\cal GDS}_A^{(\lambda)}=0$. We need to prove that for every couple $(N_{sq},\Delta)$ in such region, there exists a two-mode squeezed thermal state with number of photons $N_{sq}$ and GDS equal to $\Delta$. To do so we can consider once again only symmetric states with $\nu_1=\nu_2=\nu$. Moreover, instead of fixing $N_{sq}$ while moving $\Delta$ on a $N_{sq}$-dependent region, we can do the opposite. Indeed, for a symmetric state GDS is a monotonic function of $|r|$, which is therefore fixed by $\Delta$. The constraint given by imposing a total number of photons $N_{sq}$ then reads:
\begin{equation}
\nu\cosh(2r) -1 = N_{sq},
\end{equation}
that can be solved for $\nu\geq 1$ only if $N_{sq}$ is big enough. It is easy to see that such extreme value for $N_{sq}$ exactly coincides with the one associated with pure two-mode squeezed states.
 \end{itemize}
\end{proof}
\end{document}